\documentclass{tlp}

\submitted{24 September 2014}
\revised{9 February 2015}
\accepted{12 March 2015}

\usepackage{stmaryrd}
\usepackage{calc}
\usepackage{amssymb}
\usepackage{amsmath}
\usepackage{tikz}
\usetikzlibrary{shapes}
\usepackage{xspace}
\usepackage{url}
\usepackage{hyphenat}

\ifx\APPENDIX\undefined
\else
	\usepackage{xr}
\fi

\newcommand{\qed}{\hspace*{1em}\hbox{\proofbox}}

\newtheorem{theorem}{Theorem}[section]
\newtheorem{corollary}[theorem]{Corollary}
\newtheorem{lemma}[theorem]{Lemma}
\newtheorem{definition}[theorem]{Definition}
\newtheorem{example}[theorem]{Example}

\newcommand\TTTT{%
 \textsf{T\kern-0.2em\raisebox{-0.3em}T\kern-0.2emT\kern-0.2em%
 \raisebox{-0.3em}2}\xspace%
}
\newcommand\mkbtt{\textsf{mkbtt}\xspace}

\newcommand{\tf}[1]{{\triangledown_{\!#1}}}

\newcommand{\embsm}[1]{\vartriangleright^{\!#1}_{\!\mathsf{emb}}}
\newcommand{\m}[1]{\mathsf{#1}}
\newcommand{\mc}[1]{\mathcal{#1}}
\newcommand{\mr}[1]{\mathrm{#1}}
\newcommand{\rt}{\m{root}}
\newcommand{\Wt}{\m{w,root}}
\renewcommand{\Wt}{\m{kv}}
\newcommand{\kvc}{{\m{kv'}}}
\newcommand{\lex}{\m{lex}}
\newcommand{\mul}{\m{mul}}

\newcommand{\FF}{\mc{F}}
\newcommand{\VV}{\mc{V}}
\newcommand{\TT}{\mc{T}}
\newcommand{\Var}{\VV\mr{ar}}
\newcommand{\AC}{\mr{\m{AC}}}
\newcommand{\ackbo}{\mr{\m{ACKBO}}}
\newcommand{\scoeff}{\mathit{sc}}
\newcommand{\actkbo}{\mr{\m{ACKBO}}^\scoeff}
\newcommand{\steinbach}{\mr{\m{S}}}
\newcommand{\acrpo}{\mr{\m{ACRPO}}}
\newcommand{\acrpoo}{\mr{\m{ACRPO'}}}
\newcommand{\KV}{\mr{\m{KV}}}
\newcommand{\KVC}{\mr{\m{KV'}}}
\newcommand{\seq}[2][n]{{#2_1},\dots,{#2_{#1}}}
\newcommand{\rr}[3][f]{{#2}{\restriction}^{#3}_{#1}}
\newcommand{\rrs}[3][f]{{#2}{\restriction}^{\smash{#3}}_{#1}}
\newcommand{\RR}{\mc{R}}
\newcommand{\Nat}{\mathbb{N}}
\newcommand{\vcoeff}{\m{vc}}
\renewcommand{\AA}{\mc A}

\newcommand{\proper}{proper}
\renewcommand{\proper}{strict}

\newcommand{\High}{\m{a}}
\newcommand{\Low}{\m{b}}
\newcommand{\Bot}{\m{c}}
\newcommand{\Top}{\m{d}}
\newcommand{\ToP}{\m{e}}

\newcommand{\GT}{\mathrel{\succ}}
\newcommand{\GS}{\mathrel{\succsim}}
\newcommand{\NGT}{\mathrel{\nsucc}}
\newcommand{\REL}{\mathrel{R}}

\author[Akihisa Yamada et al.]{%
AKIHISA YAMADA\\
Research Institute for Secure Systems, AIST, Japan
\and SARAH WINKLER\\
Institute of Computer Science, University of Innsbruck, Austria
\and NAO HIROKAWA\\
School of Information Science, JAIST, Japan 
\and AART MIDDELDORP\\
Institute of Computer Science, University of Innsbruck, Austria
}

\ifx\APPENDIX\undefined
	\title[AC-KBO Revisited]{AC-KBO Revisited%
	\thanks{The research described in this paper is supported by
	the Austrian Science Fund (FWF) international project I963,
	the bilateral programs of the Japan Society for the Promotion of Science
	and the KAKENHI Grant No.\ 25730004.}
	\thanks{This is an extended version of a paper presented at the Twelfth 
	International Symposium on Functional and Logic Programming (FLOPS 2014), 
	invited as a rapid publication in TPLP. The authors acknowledge the assistance 
	of the conference chairs Michael Codish and Eijiro Sumii.}}
\else
	\title[Online appendix]{{\large\textnormal{Online appendix for the paper}}   \\
	AC-KBO revisited
	\\
	{\large\textnormal{published in Theory and Practice of Logic Programming}}
	}
\fi

\begin{document}
\maketitle
\ifx\APPENDIX\undefined
\ifx\ARTICLE\undefined
\noindent
{\bf Note:} This article has been accepted for publication in
\emph{Theory and Practice of Logic Programming}, \copyright\,Cambridge
University Press.\\
\enlargethispage{8ex}
\fi
\fi

\ifx\APPENDIX\undefined
\begin{abstract}
Equational theories that contain axioms expressing associativity and
commutativity (AC) of certain operators are ubiquitous. Theorem proving
methods in such theories rely on well-founded orders that are compatible
with the AC axioms.
In this paper we
consider various definitions of AC-compatible Knuth-Bendix orders.
The orders of Steinbach and of Korovin and Voronkov are revisited. The
former is enhanced to a more powerful version, and we modify
the latter to amend its lack of monotonicity on non-ground terms.
We further present new complexity results.
An extension reflecting the recent proposal of subterm coefficients
in standard Knuth-Bendix orders is also given.
The various orders are
compared on problems in termination and completion.
\end{abstract}


\section{Introduction}

Associative and commutative (AC) operators appear in many applications,
e.g.\ in automated reasoning with respect to algebraic structures 
such as commutative groups or rings.
We are interested in proving termination of term rewrite systems with AC
symbols. AC termination is important when deciding validity in equational
theories with AC operators by means of completion.

Several termination methods for plain rewriting have been
extended to deal with AC symbols. \citeN{BL87}
presented a characterization of polynomial interpretations that ensures
compatibility with the AC axioms. There have been numerous
papers on extending the recursive path order (RPO) of
\citeN{D82} to deal with AC symbols,
starting with the associative path order
of \citeN{BP85} and
culminating in the fully syntactic
AC-RPO of \citeN{R02}.
Several authors~\cite{KT01,MU04,GK01,ALM10} adapted
the influential dependency pair
method of \citeN{AG00} to AC rewriting.

We are aware of only two papers on AC extensions of the order (KBO) of 
\citeN{KB70}.
In this paper we revisit these orders and
present yet another AC-compatible KBO.
\citeN{S90} presented
a first version, which comes with the restriction that AC symbols are
minimal in the precedence. By incorporating ideas of \cite{R02},
\citeN{KV03b} presented a version without this restriction.
Actually, they present two versions. One is defined on ground terms
and another one on arbitrary terms. For (automatically) proving
AC termination of rewrite systems,
an AC-compatible order on arbitrary terms is required.%
\footnote{Any AC-compatible reduction order $\GT_\mathrm{g}$ on ground
terms can trivially be extended to
arbitrary terms by defining $s \GT t$ if and only if
$s\sigma \GT_\mathrm{g} t\sigma$ for all grounding substitutions $\sigma$.
This is, however, only of (mild) theoretical interest.}
We show that the second order of \citeANP{KV03b} lacks the
monotonicity property which is required by the definition of
simplification orders.
Nevertheless we prove that the order is sound for
proving termination by extending it to an AC-compatible
simplification order.
We furthermore present a simpler variant of this latter order which
properly extends the order of~\citeN{S90}.
In particular, Steinbach's order is a correct
AC-compatible simplification order, contrary to what is claimed
in~\cite{KV03b}.
We also present new complexity results which confirm that
AC rewriting is much more involved than plain rewriting.
Apart from these theoretical contributions, we implemented the
various AC-compatible KBOs to compare them also experimentally.

The remainder of this paper is organized as follows. After
recalling basic concepts of rewriting modulo AC and orders,
we revisit Steinbach's order in Section~\ref{Steinbach}.
Section~\ref{Korovin and Voronkov} is devoted to the two orders of
Korovin and Voronkov. We present a first version of our AC-compatible KBO
in Section~\ref{AC-KBO}, 
also giving the non-trivial proof that
it has the required properties. (The proofs in~\cite{KV03b} are limited
to the order on ground terms.)
In Section~\ref{complexity} we consider the complexity of the
membership and orientation decision problems for the various orders.
In Section~\ref{AC-RPO} we compare 
AC-KBO with AC-RPO.
In Section~\ref{subterm coefficients} our order is strengthened
with subterm coefficients.
In order to show effectiveness of these orders
experimental data is provided in Section~\ref{experiments}.
The paper is concluded in Section~\ref{conclusion}.

This article is an updated and extended version of~\cite{YWHM14}.
Our earlier results on complexity are extended by showing that
the orientability problems for different versions of AC-KBO
are in NP. Moreover, we include a comparison with AC-RPO, which
we present in a slightly simplified manner compared to \cite{R02}.
Due to space limitations, some proofs can be found in the online appendix.

\section{Preliminaries}

We assume familiarity with rewriting and termination. Throughout this
paper we deal with rewrite systems over
a set $\VV$ of variables and a \emph{finite} signature $\FF$
together with a designated subset $\FF_\AC$ of binary AC symbols. The
congruence relation induced by the equations $f(x,y) \approx f(y,x)$ and
$f(f(x,y),z) \approx f(x,f(y,z))$ for all $f \in \FF_\AC$ is denoted by
$=_\AC$. A term rewrite system (TRS for short) $\RR$ is AC terminating if
the relation ${=_\AC} \cdot {\to_\RR} \cdot {=_\AC}$ is well-founded. In
this paper AC termination is established by
\emph{AC-compatible simplification orders} $\GT$, which are
\proper\ orders (i.e., irreflexive and transitive relations)
closed under contexts and substitutions that have the subterm property
$f(\seq{t}) \GT t_i$ for all $1 \leqslant i \leqslant n$ and satisfy
${=_\AC} \cdot {\GT} \cdot {=_\AC} \subseteq {\GT}$.
A \proper\ order $\GT$ is \emph{AC-total} if
$s \GT t$, $t \GT s$ or $s =_\AC t$, for all ground terms $s$ and $t$.
A pair $({\GS},{\GT})$ consisting of a preorder $\GS$ and a \proper\
order $\GT$ is said to be an \emph{order pair} if the \emph{compatibility}
condition ${\GS \cdot \GT \cdot \GS} \subseteq {\GT}$ holds.

\begin{definition}
\label{lex and mul}
Let $\GT$ be a \proper\ order and $\GS$ be a preorder on a set $A$.
The \emph{lexicographic extensions} $\GT^\lex$ and $\GS^\lex$
are defined as follows:
\begin{itemize}
\item
$\vec{x} \GS^\lex \vec{y}$ if $\vec{x} \sqsupset_k^\lex \vec{y}$
for some $1 \leqslant k \leqslant n$,
\item
\smallskip
$\vec{x} \GT^\lex \vec{y}$ if $\vec{x} \sqsupset_k^\lex \vec{y}$
for some $1 \leqslant k < n$.
\end{itemize}
Here $\vec{x} = (\seq{x})$, $\vec{y} = (\seq{y})$, and
$\vec{x} \sqsupset_k^\lex \vec{y}$ denotes the following condition:
$x_i \GS y_i$ for all $i \leqslant k$ and either 
$k < n$ and $x_{k+1} \GT y_{k+1}$ or $k = n$.
The \emph{multiset extensions} $\GT^\mul$ and $\GS^\mul$ are
defined as follows:
\begin{itemize}
\item
$M \GS^\mul N$ if $M \sqsupset_k^\mul N$ for some
$0 \leqslant k \leqslant \min(m,n)$,
\item
\smallskip
$M \GT^\mul N$ if $M \sqsupset_k^\mul N$ for some
$0 \leqslant k \leqslant \min(m-1,n)$.
\end{itemize}
Here $M \sqsupset_k^\mul N$ if $M$ and $N$ consist of $\seq[m]{x}$ 
and $\seq{y}$ respectively such that $x_j \GS y_j$ for all 
$j \leqslant k$, and for every $k < j \leqslant n$ there is some
$k < i \leqslant m$ with $x_i \GT y_j$.
\end{definition}

Note that these extended relations depend on both $\GS$ and $\GT$.
The following result is folklore;
a recent formalization of multiset
extensions in Isabelle/HOL is presented in \cite{TAN12}.

\begin{theorem}
\label{thm:order pair}
If $(\GS,\GT)$ is an order pair then $(\GS^\lex,\GT^\lex)$
and $(\GS^\mul,\GT^\mul)$ are order pairs.
\qed
\end{theorem}

\section{Steinbach's Order}
\label{Steinbach}

In this section we recall the AC-compatible KBO $>_\steinbach$ of
\citeN{S90}, which reduces to the standard KBO if AC symbols are
absent.%
\footnote{The version in \cite{S90} is slightly more general, since non-AC
function symbols can have arbitrary status. To simplify the discussion, we
do not consider status in this paper.}
The order $>_\steinbach$ depends on a precedence
and an admissible weight function.
A \emph{precedence} $>$ is a \proper\ order on $\FF$.
A \emph{weight function} $(w,w_0)$ for a signature $\FF$ consists of
a mapping $w\colon \FF \to \Nat$ and a constant $w_0 > 0$ such that 
$w(c) \geqslant w_0$ for every constant $c \in \FF$.
The \emph{weight} of a term $t$ is recursively computed as follows:
\[
w(t) = 
\begin{cases}
w_0 & \text{if $t \in \VV$} \\
\displaystyle w(f) + \smash[b]{\sum_{1 \leqslant i \leqslant n}} w(t_i) &
\text{if $t = f(\seq{t})$}
\end{cases}
\]

\bigskip

\noindent
A weight function $(w,w_0)$ is \emph{admissible} for $>$ if 
every unary $f$ with $w(f) = 0$ satisfies $f > g$ for all
function symbols $g$ different from $f$.
Throughout this paper we assume admissibility.


The \emph{top\hyp flattening} \cite{R02}
of a term $t$ with respect to an AC symbol $f$
is the multiset $\tf{f}(t)$ defined inductively as follows:
\[
\tf{f}(t) =
\begin{cases}
\{ t \} & \text{if $\rt(t) \neq f$} \\
\tf{f}(t_1) \uplus \tf{f}(t_2) & \text{if $t = f(t_1,t_2)$}
\end{cases}
\]

\begin{definition}
\label{def:Steinbach}
Let $>$ be a precedence and $(w,w_0)$ a weight function.
The order $>_\steinbach$ is inductively defined
as follows: $s >_\steinbach t$ if
$|s|_x \geqslant |t|_x$ for all $x \in \VV$ and either $w(s) > w(t)$, or
$w(s) = w(t)$ and one of the following alternatives holds:
\begin{enumerate}
\setcounter{enumi}{-1}
\item
$s = f^k(t)$ and $t \in \VV$ for some $k > 0$,
\smallskip
\item
$s = f(\seq{s})$, $t = g(\seq[m]{t})$, and $f > g$,
\smallskip
\item
$s = f(\seq{s})$, $t = f(\seq{t})$, $f \notin \FF_\AC$,
$(\seq{s}) >_\steinbach^\lex (\seq{t})$,
\smallskip
\item
$s = f(s_1,s_2)$, $t = f(t_1,t_2)$, $f \in \FF_\AC$,
and $\tf{f}(s) >_\steinbach^\mul \tf{f}(t)$.
\end{enumerate}
\smallskip
The relation $=_\AC$ is used as preorder in
$>_\steinbach^\lex$ and $>_\steinbach^\mul$.
\end{definition}

Cases~0--2 are the same as in the standard Knuth-Bendix order.
In case~3 terms rooted by the same AC symbol $f$ are treated by
comparing their top\hyp flattenings in the multiset extension of
$>_\steinbach$.

\begin{example}
\label{example steinbach}
Consider the signature $\FF = \{ \m{a}, \m{f}, + \}$ with
${+} \in \FF_\AC$, precedence
$\m{f} > \m{a} > +$ 
and admissible weight function $(w,w_0)$ with $w(\m{f}) = w(+) = 0$
and $w_0 = w(\m{a}) = 1$. Let $\RR_1$ be the following ground TRS:
\\[-.5ex]
\begin{minipage}{.45\textwidth}
\begin{align}
\m{f}(\m{a} + \m{a}) &\to \m{f}(\m{a}) + \m{f}(\m{a})
 \tag{1}
\end{align}
\end{minipage}
\hfill
\begin{minipage}{.45\textwidth}
\begin{align}
\m{a} + \m{f}(\m{f}(\m{a})) &\to \m{f}(\m{a}) + \m{f}(\m{a})
 \tag{2}
\end{align}
\end{minipage}

\medskip

\noindent
For $1 \leqslant i \leqslant 2$, let $\ell_i$ and $r_i$ be the
left- and right-hand side of rule $(i)$,
$S_i = \tf{+}(\ell_i)$ and $T_i = \tf{+}(r_i)$.
Both rules vacuously satisfy the variable condition.
We have $w(\ell_1) = 2 = w(r_1)$ and $\m{f} > +$, so
$\ell_1 >_\steinbach r_1$ holds by case~1.
We have $w(\ell_2) = 2 = w(r_2)$,
$S_2 = \{ \m{a}, \m{f}(\m{f}(\m{a})) \}$, 
and $T_2 = \{ \m{f}(\m{a}), \m{f}(\m{a}) \}$. Since
$\m{f}(\m{a}) >_\steinbach \m{a}$ holds by case~1,
$\m{f}(\m{f}(\m{a})) >_\steinbach \m{f}(\m{a})$ holds
by case~2, and therefore $\ell_2 >_\steinbach r_2$ by case~3.
\end{example}

\begin{theorem}[\citeNP{S90}]
\label{thm:Steinbach}
If every symbol in $\FF_\AC$ is minimal with respect to $>$
then $>_\steinbach$ is an AC-compatible simplification order.%
\footnote{In \cite{S90} AC symbols are further required to have
weight $0$ because terms are flattened.
Our version of $>_\steinbach$ does not impose
this restriction due to the use of top\hyp flattening.}
\end{theorem}

In Section~\ref{AC-KBO} we reprove\footnote{%
The counterexample in \cite{KV03b} against the monotonicity of
$>_\steinbach$ is invalid as the condition that AC symbols are
\emph{minimal} in the precedence is not satisfied.}
Theorem~\ref{thm:Steinbach}
by showing that $>_\steinbach$ is a special case of
our new AC-compatible Knuth-Bendix order.

\section{Korovin and Voronkov's Orders}
\label{Korovin and Voronkov}

In this section we recall the orders of~\citeN{KV03b}.
The first one is defined on ground terms.
The difference
with $>_\steinbach$ is that in case~3 of the definition a
further case analysis is performed based on 
terms in $S$ and $T$ whose root symbols are
not smaller than $f$ in the precedence.
Rather than recursively comparing these terms with the order being
defined, a lighter non-recursive version is used in which the weights and
root symbols are considered. This is formally defined below.

Given a multiset $T$ of terms, a function symbol $f$, and a binary relation
$R$ on function symbols, we define the following submultisets of $T$: 
\begin{align*}
T{\restriction}_\VV &= \{ x \in T \mid x \in \VV \} &
\rr{T}{R} &= \{ t \in T \setminus \VV \mid \rt(t) \mathrel{R} f \}
\end{align*}

\begin{definition}
\label{def:KV ground}
Let $>$ be a precedence and $(w,w_0)$ a weight function.%
\footnote{Here we do not impose totality on precedences, cf.\ \cite{KV03b}.
See also Example~\ref{example partiality}.}
First we define the auxiliary relations $=_\Wt$ and $>_\Wt$
on ground terms as follows:
\begin{itemize}
\item
$s =_\Wt t$ if
$w(s) = w(t)$ and $\rt(s) = \rt(t)$,
\item
$s >_\Wt t$ if
either $w(s) > w(t)$ or both $w(s) = w(t)$ and $\rt(s) > \rt(t)$.
\end{itemize}
The order $>_\KV$ is inductively defined on ground terms as follows:
$s >_\KV t$ if either $w(s) > w(t)$, or
$w(s) = w(t)$ and one of the following alternatives holds:
\begin{enumerate}
\item
$s = f(\seq{s})$, $t = g(\seq[m]{t})$, and $f > g$,
\smallskip
\item
$s = f(\seq{s})$, $t = f(\seq{t})$, $f \notin \FF_\AC$,
$(\seq{s}) >_\KV^\lex (\seq{t})$,
\item
\smallskip
$s = f(s_1,s_2)$, $t = f(t_1,t_2)$, $f \in \FF_\AC$,
and for $S = \tf{f}(s)$ and $T = \tf{f}(t)$
\begin{itemize}
\item[(a)]
\smallskip
$\rrs{S}{\nless} >_\Wt^\mul \rrs{T}{\nless}$,
or
\item[(b)]
\smallskip
$\rrs{S}{\nless} =_\Wt^\mul \rrs{T}{\nless}$ and $|S| > |T|$, or
\item[(c)]
\smallskip
$\rrs{S}{\nless} =_\Wt^\mul \rrs{T}{\nless}$, $|S| = |T|$, and
$S >_\KV^\mul T$.
\end{itemize}
\end{enumerate}
Here $=_\AC$ is used as preorder in $>_\KV^\lex$ and $>_\KV^\mul$ whereas
$=_\Wt$ is used in $>_\Wt^\mul$.
\end{definition}

Only in cases~2 and 3(c)
the order $>_\KV$ is used recursively.
In case~3 terms rooted by the same AC symbol $f$ are compared by
extracting from the top\hyp flattenings $S$ and $T$ the multisets
$\rrs{S}{\nless}$ and $\rrs{T}{\nless}$
consisting of all terms rooted by a function symbol not
smaller than $f$ in the precedence.
If $\rrs{S}{\nless}$ is larger than
$\rrs{T}{\nless}$ in the multiset extension of $>_\Wt$,
we conclude in case~3(a). Otherwise the multisets must be equal
(with respect to $=_\Wt^\mul$).
If $S$ has more terms than $T$, we conclude in case~3(b). In the
final case~3(c) $S$ and $T$ have the same number of terms and we
compare $S$ and $T$ in the multiset extension of $>_\KV$.

\begin{theorem}[\citeNP{KV03b}]
\label{thm:KV ground}
The order $>_\KV$ is an AC-compatible simplification order on ground
terms. If $>$ is total then $>_\KV$ is AC-total on ground terms.
\end{theorem}

The two orders $>_\KV$ and $>_\steinbach$ are incomparable on ground TRSs.

\begin{example}
\label{KV does not subsume S}
Consider again the ground TRS $\RR_1$ of Example~\ref{example steinbach}.
To orient rule (1) with $>_\KV$, the weight of the unary function
symbol $\m{f}$ must be $0$ and admissibility demands $\m{f} > \m{a}$
and $\m{f} > +$. Hence rule (1) is handled by case~1 of the
definition. For rule (2), the multisets
$S = \{ \m{a}, \m{f}(\m{f}(\m{a})) \}$ and
$T = \{ \m{f}(\m{a}), \m{f}(\m{a}) \}$ are compared in case~3. We have
$\rrs[+]{S}{\nless} = \{ \m{f}(\m{f}(\m{a})) \}$ if $+ > \m{a}$
and $\rrs[+]{S}{\nless} = S$ otherwise.
In both cases we have $\rrs[+]{T}{\nless} = T$.
Note that neither
$\m{a} >_\Wt \m{f}(\m{a})$ nor
$\m{f}(\m{f}(\m{a})) >_\Wt \m{f}(\m{a})$ holds. Hence case~3(a)
does not apply. But also cases~3(b) and~3(c) are not applicable as
$\m{f}(\m{f}(\m{a})) =_\Wt \m{f}(\m{a})$ and
$\m{a} \neq_\Wt \m{f}(\m{a})$. Hence, independent of the choice of
$>$, $\RR_1$ cannot be proved terminating by $>_\KV$.
Conversely, the TRS $\RR_2$ resulting from reversing rule (2)
in $\RR_1$ can be proved terminating by $>_\KV$ but not by
$>_\steinbach$.
\end{example}


Next we present the second order of \citeN{KV03b}, the extension of
$>_\KV$ to non-ground terms. Since it coincides with $>_\KV$
on ground terms, we use the same notation for the order.

In case~3 of the following definition, also variables appearing
in the top\hyp flattenings $S$ and $T$ are taken into account in the
first multiset comparison.
Given a relation $\REL$ on terms, we write 
$S \REL^f T$ for
\[
\rrs{S}{\nless}\,\REL^\mul\,\rrs{T}{\nless}
\uplus T{\restriction}_\VV - S{\restriction}_\VV
\]
Note that $\REL^f$ depends on a precedence $>$. Whenever we
use $\REL^f$, $>$ is defined.

\begin{definition}
\label{def:kv}
Let $>$ be a precedence and $(w,w_0)$ a weight function.
The orders $=_\Wt$ and $>_\Wt$ are extended to non-ground terms
as follows:
\begin{itemize}
\item
$s =_\Wt t$ if $|s|_x = |t|_x$ for all $x \in \VV$,
$w(s) = w(t)$ and $\rt(s) = \rt(t)$,
\item
$s >_\Wt t$ if $|s|_x \geqslant |t|_x$ for all $x \in \VV$ and
either $w(s) > w(t)$ or both $w(s) = w(t)$ and $\rt(s) > \rt(t)$.
\end{itemize}
\end{definition}

Some tricky features of the relations $=_\Wt$ and $>_\Wt$ are
illustrated below.

\begin{example}
\label{counterexample preliminary}
Let $\m{c}$ be a constant and $\m{f}$ a unary symbol.
We have $\m{f}(\m{c}) >_\Wt \m{c}$
whenever admissibility is assumed:
If $w(\m{f}) > 0$ then $w(\m{f}(\m{c})) > w(\m{c})$,
and if $w(\m{f}) = 0$ then admissibility imposes $\m{f} > \m{c}$.
On the other hand, $\m{f}(x) >_\Wt x$ holds only if $w(\m{f}) > 0$,
since $\m{f} \ngtr x$.
Furthermore, $\m{f}(x) =_\Wt x$ does not hold as
$\m{f} \neq x$.
\end{example}

\begin{example}
\label{counterexample2 preliminary}
Let $\m{c}$ be a constant with $w(\m{c}) = w_0$, $\m{f}$ a unary symbol,
and $\m{g}$ a non-AC binary symbol.
We do not have
$\ell = \m{g}(\m{f}(\m{c}),x) >_\Wt \m{g}(\m{c},\m{f}(\m{c})) = r$
since $w(\ell) = w(r)$ and $\rt(\ell) = \rt(r) = \m{g}$.
On the other hand, $\ell =_\Wt r$ also does not hold
since the condition ``$|s|_x = |t|_x$ for all $x \in \VV$\,'' is not
satisfied.
\end{example}

Now the non-ground version of $>_\KV$ is defined as follows.

\begin{definition}
\label{def:KV}
Let $>$ be a precedence and $(w,w_0)$ a weight function.
The order $>_\KV$ is inductively defined as follows: $s >_\KV t$ if
$|s|_x \geqslant |t|_x$ for all $x \in \VV$ and either $w(s) > w(t)$, or
$w(s) = w(t)$ and one of the following alternatives holds:
\begin{enumerate}
\setcounter{enumi}{-1}
\item
$s = f^k(t)$ and $t \in \VV$ for some $k > 0$,
\smallskip
\item
$s = f(\seq{s})$, $t = g(\seq[m]{t})$, and $f > g$,
\smallskip
\item
$s = f(\seq{s})$, $t = f(\seq{t})$, $f \notin \FF_\AC$,
$(\seq{s}) >_\KV^\lex (\seq{t})$,
\item
\smallskip
$s = f(s_1,s_2)$, $t = f(t_1,t_2)$, $f \in \FF_\AC$,
and for $S = \tf{f}(s)$ and $T = \tf{f}(t)$
\begin{itemize}
\item[(a)]
\smallskip
$S >_\Wt^f T$,
or
\item[(b)]
\smallskip
$S =_\Wt^f T$ and $|S| > |T|$, or
\item[(c)]
\smallskip
$S =_\Wt^f T$, $|S| = |T|$, and $S >_\KV^\mul T$.
\end{itemize}
\end{enumerate}
Here $=_\AC$ is used as preorder in $>_\KV^\lex$ and $>_\KV^\mul$
whereas $=_\Wt$ is used in $>_\Wt^\mul$.
\end{definition}

Contrary to what is claimed in \cite{KV03b}, the order $>_\KV$
of Definition~\ref{def:KV}
is not a simplification order because it lacks the monotonicity property
(i.e., $>_\KV$ is not closed under contexts),
as shown in the following examples.

\begin{example}
\label{counterexample}
We continue Example~\ref{counterexample preliminary}
by adding an AC symbol $+$.
We obviously have $\m{f}(x) >_\KV x$. However,
$\m{f}(x) + y >_\KV x + y$ does not hold if $w(\m{f}) = 0$.
Let 
\begin{align*}
S &= \tf{+}(s) = \{ \m{f}(x), y \} &
T &= \tf{+}(t) = \{ x, y \}
\end{align*}
We have
$\rrs[+]{S}{\nless} = \{ \m{f}(x) \}$,
and $\rrs[+]{T}{\nless} \cup T{\restriction}_\VV - S{\restriction}_\VV =
\{ x \}$.
As shown in Example~\ref{counterexample preliminary},
neither $\m{f}(x) >_\Wt x$ nor $\m{f}(x) =_\Wt x$ holds.
Hence none of the cases~3(a,b,c) of Definition~\ref{def:KV} can be applied.
\end{example}

Note that the use of a unary function of weight 0 is not crucial.
The following example illustrates that
the non-ground version of $>_\KV$ need not be closed under
contexts, even if there is no unary symbol of weight zero.

\begin{example}
\label{counterexample 2}
We continue Example~\ref{counterexample2 preliminary}
by adding an AC symbol $+$ with $\m{g} > + > \m{c}$. We have
\[
\ell = \m{g}(\m{f}(\m{c}),x) >_\KV \m{g}(\m{c},\m{f}(\m{c})) = r
\]
by case~2. However,
$s = \ell + \m{c} >_\KV r + \m{c} = t$ does not hold.
Let
\begin{align*}
S &= \tf{+}(s) = \{ \ell, \m{c} \} &
T &= \tf{+}(t) = \{ r, \m{c} \}
\end{align*}
We have
$\rrs[+]{S}{\nless} = \{ \ell \}$, $\rrs[+]{T}{\nless} = \{ r \}$, and
$S{\restriction}_\VV = T{\restriction}_\VV = \varnothing$.
As shown in Example~\ref{counterexample2 preliminary},
$\ell >_\Wt r$ does not hold.
Hence case~3(a) in Definition~\ref{def:KV} does not apply.
But also $\ell =_\Wt r$ does not hold, excluding
3(b) and 3(c).
\end{example}

These examples do not refute the soundness of $>_\KV$ for proving AC
termination; note that e.g.\ in Example \ref{counterexample}
also $x + y >_\KV \m{f}(x) + y$ does not hold. We prove soundness by
extending $>_\KV$ to $>_\KVC$ which has all desired properties.

%

\begin{definition}
\label{def:KV'}
The order $>_\KVC$ is obtained as in Definition~\ref{def:KV} after
replacing
$=_\Wt^f$ by $\geqslant_\kvc^f$ in cases~3(b) and 3(c),
and using $\geqslant_\kvc$ as preorder in $>_\Wt^\mul$ in case~3(a).
Here the relation $\geqslant_\kvc$ is defined as follows:
\begin{itemize}
\item
$s \geqslant_\kvc t$ if $|s|_x \geqslant |t|_x$ for all $x \in \VV$
and either $w(s) > w(t)$, or $w(s) = w(t)$ and
either $\rt(s) \geqslant \rt(t)$ or $t \in \VV$.
\end{itemize}
\end{definition}

Note that $\geqslant_\kvc$ is a preorder that contains $=_\AC$.

\begin{example}
Consider again Example~\ref{counterexample}.
We have $\m{f}(x) \geqslant_\kvc x$ due to the new possibility
``$t \in \VV$\,''.
We have
$\m{f}(x) + y >_\KVC x + y$ because now case~3(c) applies:
$\rrs[+]{S}{\nless} = \{ \m{f}(x) \} \geqslant_\kvc^\mul
\{ x \} = \rrs[+]{T}{\nless} \uplus T{\restriction}_\VV -
S{\restriction}_\VV$, $|S| = 2 = |T|$, and
$S = \{ \m{f}(x), y \} >_\KVC^\mul \{ x, y \} = T$ because
$\m{f}(x) >_\KVC x$.
Analogously, we have $\ell + \m{c} >_\KVC r + \m{c}$ for 
Example~\ref{counterexample 2}.
\end{example}

%

The proof of the following result can be found in the online appendix.

\begin{theorem}
\label{KV' correctness}
The order $>_\KVC$ is an AC-compatible simplification order. 
\end{theorem}

Since the inclusion ${>_\KV} \subseteq {>_\KVC}$ obviously holds,
it follows that $>_\KV$ is a sound method for establishing AC termination,
despite the lack of monotonicity.

\section{AC-KBO}
\label{AC-KBO}

In this section we present another AC-compatible simplification order. In
contrast to $>_\KVC$,
our new order $>_\ackbo$ contains $>_\steinbach$. Moreover, its
definition is simpler than $>_\KVC$ since we avoid the use of an auxiliary
order in case~3.
In the next section we show that
$>_\ackbo$ is decidable in polynomial-time, whereas the membership
decision problem for $>_\KVC$ is NP-complete.
Hence it will be used as the basis for the extension discussed in
Section~\ref{subterm coefficients}.

\begin{definition}
\label{def:ackbo}
Let $>$ be a precedence and $(w,w_0)$ a weight function.
We define $>_\ackbo$ inductively
as follows: $s >_\ackbo t$ if
$|s|_x \geqslant |t|_x$ for all $x \in \VV$ and either $w(s) > w(t)$, or
$w(s) = w(t)$ and one of the following alternatives holds:
\begin{enumerate}
\setcounter{enumi}{-1}
\item
$s = f^k(t)$ and $t \in \VV$ for some $k > 0$,
\smallskip
\item
$s = f(\seq{s})$, $t = g(\seq[m]{t})$, and $f > g$,
\smallskip
\item
$s = f(\seq{s})$, $t = f(\seq{t})$, $f \notin \FF_\AC$,
$(\seq{s}) >_\ackbo^\lex (\seq{t})$,
\item
\smallskip
$s = f(s_1,s_2)$, $t = f(t_1,t_2)$, $f \in \FF_\AC$,
and for $S = \tf{f}(s)$ and $T = \tf{f}(t)$
\begin{itemize}
\item[(a)]
\smallskip
$S >_\ackbo^f T$,
or
\item[(b)]
\smallskip
$S =_\AC^f T$, and $|S| > |T|$, or
\item[(c)]
\smallskip
$S =_\AC^f T$, $|S| = |T|$, and
$\rrs{S}{<} >_\ackbo^\mul \rrs{T}{<}$.
\end{itemize}
\end{enumerate}
\smallskip
The relation $=_\AC$ is used as preorder in
$>_\ackbo^\lex$ and $>_\ackbo^\mul$.
\end{definition}

Note that, in contrast to $>_\KV$, in case~3(c) we compare the multisets
$\rrs{S}{<}$ and $\rrs{T}{<}$ rather than
$S$ and $T$ in the multiset extension of $>_\ackbo$.

Steinbach's order is a special case of the order defined above.

\begin{theorem}
\label{S vs ACKBO}
If every AC symbol has minimal precedence then
${>_\steinbach} = {>_\ackbo}$.
\end{theorem}
\begin{proof*}
Suppose that
every function symbol in $\FF_\AC$ is minimal with respect to $>$.
We show that $s >_\steinbach t$ if and only if $s >_\ackbo t$
by induction on $s$. It is clearly sufficient to consider 
case~3 in Definition~\ref{def:Steinbach}
and cases~3(a,b,c) in Definition~\ref{def:ackbo}.
So let $s = f(s_1,s_2)$ and $t = f(t_1,t_2)$ such that
$w(s) = w(t)$ and $f \in \FF_\AC$.
Let $S = \tf{f}(s)$ and $T = \tf{f}(t)$.
\begin{itemize}
\item
Let $s >_\steinbach t$ by case~3.
We have $S >_\steinbach^\mul T$.
Since $S >_\steinbach^\mul T$ involves only comparisons
$s' >_\steinbach t'$ for subterms $s'$ of $s$,
the induction hypothesis yields $S >_\ackbo^\mul T$.
Because $f$ is minimal in $>$,
$S = \rrs{S}{\nless} \uplus S{\restriction_\VV}$
and $T = \rrs{T}{\nless} \uplus T{\restriction_\VV}$. 
For no elements
$u \in S{\restriction_\VV}$ and $v \in \rrs{T}{\nless}$,
$u >_\ackbo v$ or $u =_\AC v$ holds. Hence $S >_\ackbo^\mul T$
implies $S >_\ackbo^f T$ or both $S =_\AC^f T$
and
$S{\restriction_\VV} \supsetneq T{\restriction_\VV}$.
In the former case $s >_\ackbo t$
is due to case~3(a) in Definition~\ref{def:ackbo}. In the latter case
we have $|S| > |T|$ and $s >_\ackbo t$ follows by case 3(b).
\smallskip
\item
Let $s >_\ackbo t$ by applying one of the cases 3(a,b,c) in
Definition~\ref{def:ackbo}.
\begin{itemize}
\item
Suppose 3(a) applies. Then we have $S >_\ackbo^f T$.
Since $f$ is minimal in $>$,
$\rrs{S}{\nless} = S - S{\restriction}_\VV$ and
$\rrs{T}{\nless} \uplus T{\restriction}_\VV = T$.
Hence $S >_\ackbo^\mul (T - S{\restriction}_\VV) \uplus
S{\restriction}_\VV \supseteq T$. We obtain
$S >_\steinbach^\mul T$ from the induction hypothesis and thus
case~3 in Definition~\ref{def:Steinbach} applies.
\item
Suppose 3(b) applies. Analogous to the previous case, the inclusion
$S =_\AC^\mul (T - S{\restriction}_\VV) \uplus S{\restriction}_\VV
\supseteq T$ holds.
Since $|S| > |T|$, $S =_\AC^\mul T$ is not possible. Thus
$(T - S{\restriction}_\VV) \uplus S{\restriction}_\VV \supsetneq T$ and
hence $S >_\steinbach^\mul T$.
\item
If case~3(c) applies then $\rrs{S}{<} >_\ackbo^\mul \rrs{T}{<}$.
This is impossible since both sides are empty as
$f$ is minimal in $>$.
\qed
\end{itemize}
\end{itemize}
\end{proof*}

The following example shows that $>_\ackbo$ is a proper extension of 
$>_\steinbach$ and incomparable with $>_\KVC$.

\begin{example}
\label{AC-KBO vs KVC}
Consider the TRS $\RR_3$ consisting of the rules
\begin{xalignat*}{3}
\m{f}(x+y) &\to \m{f}(x)+y
&
\m{h}(\m{a},\m{b}) &\to \m{h}(\m{b},\m{a})
&
\m{h}(\m{g}(\m{a}),\m{a}) &\to \m{h}(\m{a},\m{g}(\m{b}))
\\
\m{g}(x)+y &\to \m{g}(x+y)
&
\m{h}(\m{a},\m{g}(\m{g}(\m{a}))) &\to \m{h}(\m{g}(\m{a}),\m{f}(\m{a}))
&
\m{h}(\m{g}(\m{a}),\m{b}) &\to \m{h}(\m{a},\m{g}(\m{a}))
\\
\m{f}(\m{a})+\m{g}(\m{b}) &\to \m{f}(\m{b})+\m{g}(\m{a})
\end{xalignat*}
over the signature $\{ {+}, \m{f}, \m{g}, \m{h}, \m{a}, \m{b} \}$ with
${+} \in \FF_\AC$. Consider the precedence
\[
\m{f} > {+} > \m{g} > \m{a} > \m{b} > \m{h}
\]
together with the admissible weight function $(w,w_0)$ with 
\begin{align*}
w({+}) &= w(\m{h}) = 0&
w(\m{f}) &= w(\m{a}) = w(\m{b}) = w_0 = 1&
w(\m{g}) &= 2
\end{align*}
The interesting rule is
$\m{f}(\m{a})+\m{g}(\m{b}) \to \m{f}(\m{b})+\m{g}(\m{a})$.
For $S = \tf{\,+}(\m{f}(\m{a})+\m{g}(\m{b}))$ and
$T = \tf{\,+}(\m{f}(\m{b})+\m{g}(\m{a}))$
the multisets
$S' = \rrs[+]{S}{\nless} = \{ \m{f}(\m{a}) \}$ and
$T' = \rrs[+]{T}{\nless} \uplus T{\restriction}_\VV - 
S{\restriction}_\VV = \{ \m{f}(\m{b}) \}$ satisfy
$S' >_\ackbo^\mul T'$ as $\m{f}(\m{a}) >_\ackbo \m{f}(\m{b})$,
so that case~3(a) of Definition~\ref{def:ackbo} applies.
All other rules are oriented from left to right by both $>_\KVC$ and
$>_\ackbo$, and they enforce a precedence
and weight function which are identical (or very similar) to the one
given above.  Since $>_\KVC$ orients the rule
$\m{f}(\m{a})+\m{g}(\m{b}) \to \m{f}(\m{b})+\m{g}(\m{a})$ from right
to left, $\RR_3$ cannot be compatible with $>_\KVC$.
It is easy to see that the rule $\m{g}(x)+y \to \m{g}(x+y)$ requires $+ >
\m{g}$, and hence $>_\steinbach$ cannot be applied.
\end{example}

\begin{figure}[tb]
\centering
\begin{tikzpicture}[baseline=(A),xscale=1.4,yscale=0.8]
\node (A) at (6,4.5) {};
\draw[semithick,rounded corners] (1,2)
 node[anchor=south west]{$>_\KVC$} rectangle (6,6);
\draw[semithick,rounded corners] (3.5,0)
 node[anchor=south west]{$>_\ackbo$} rectangle (8.5,4);
\draw[semithick] (6,2) ellipse (1.5cm and 1.25cm);
\draw (6,.75) node[anchor=south]{$>_\steinbach$};
\draw (6.7,1.75) node[anchor=south]{$\stackrel{\raisebox{1mm}{$\RR_1$}}%
 {\scriptstyle\bullet}$};
\draw (2.25,3.7) node[anchor=south]{$\stackrel{\raisebox{1mm}{$\RR_2$}}%
 {\scriptstyle\bullet}$};
\draw (8,1.75) node[anchor=south]{$\stackrel{\raisebox{1mm}{$\RR_3$}}%
 {\scriptstyle\bullet}$};
\draw (6.2,4.1) node[anchor=south west]{
\begin{minipage}{5cm}
\begin{tabular}{ll}
$\RR_1$ &
Example~\ref{example steinbach} (and \ref{KV does not subsume S})
\\[.5ex]
$\RR_2$ &
Example~\ref{KV does not subsume S}
\\[.5ex]
$\RR_3$ &
Example~\ref{AC-KBO vs KVC}
\end{tabular}
\end{minipage}
};
\end{tikzpicture}
\caption{Comparison.}
\label{comparison}
\end{figure}
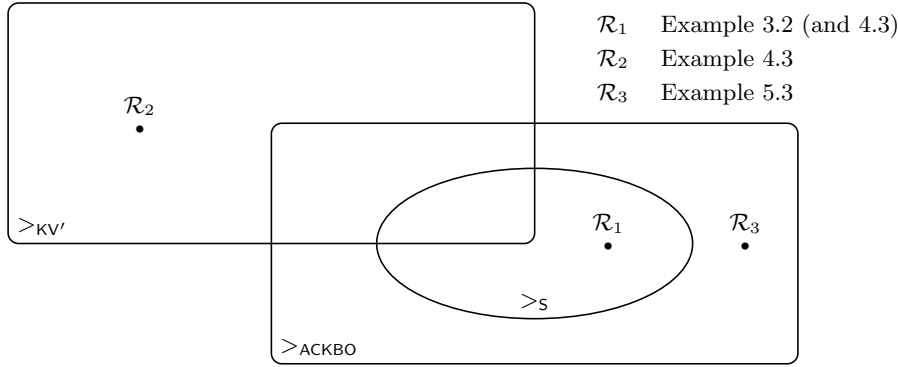
Fig.~\ref{comparison} summarizes the relationships between the orders
introduced so far.
In the following, we show that $>_\ackbo$ is an AC-compatible 
simplification order.
As a consequence, correctness of $>_\steinbach$
(i.e., Theorem~\ref{thm:Steinbach}) is concluded by
Theorem~\ref{S vs ACKBO}.

In the online appendix we prove the following property.
\begin{lemma}
\label{lem:ackbo order pair}
The pair $({=_\AC},{>_\ackbo})$ is an order pair.
\end{lemma}

The subterm property is an easy consequence of transitivity and
admissibility.

\begin{lemma}
\label{lem:subterm}
The order $>_\ackbo$ has the subterm property.
\qed
\end{lemma}

Next we prove that $>_\ackbo$ is closed under contexts.
The following lemma is an auxiliary result needed for its proof.
In order to reuse this lemma for the correctness proof of $>_\KVC$
in the online appendix,
we prove it in an abstract setting.

\begin{lemma}
\label{lem:tf}
Let $({\GS},{\GT})$ be an order pair and $f \in \FF_\AC$ with
$f(u,v) \GT u, v$ for all terms $u$ and $v$. If $s \GS t$ then
$\{ s \} \GS^\mul \tf{f}(t)$ or $\{ s \} \GT^\mul \tf{f}(t)$.
If $s \GT t$ then $\{ s \} \GT^\mul \tf{f}(t)$.
\end{lemma}

\begin{proof}
Let $\tf{f}(t) = \{ \seq[m]{t} \}$.
If $m = 1$ then $\tf{f}(t) = \{ t \}$ and the lemma holds trivially.
Otherwise we get $t \GT t_j$ for all $1 \leqslant j \leqslant m$ by
recursively applying the assumption. Hence $s \GT t_j$
by the transitivity of $\GT$ or the compatibility of $\GT$ and $\GS$.
We conclude that $\{ s \} \GT^\mul \tf{f}(t)$.
\end{proof}

In the following proof of closure under contexts, admissibility is
essential. This is in contrast to the corresponding result for
standard KBO.

\begin{lemma}
\label{lem:ackbo monotone}
If $(w,w_0)$ is admissible for $>$ then $>_\ackbo$ is closed under
contexts.
\end{lemma}

\begin{proof}
Suppose $s >_\ackbo t$. We consider the context $h(\Box,u)$ with
$h \in \FF_\AC$ and $u$ an arbitrary term, and prove that
$s' = h(s,u) >_\ackbo h(t,u) = t'$. Closure under contexts of
$>_\ackbo$ follows
then by induction; contexts rooted by a non-AC symbol are handled as in
the proof for standard KBO.

If $w(s) > w(t)$ then obviously $w(s') > w(t')$. So we assume $w(s) =
w(t)$. Let $S = \tf{h}(s)$, $T = \tf{h}(t)$, and $U = \tf{h}(u)$.  Note
that $\tf{h}(s') = S \uplus U$ and $\tf{h}(t') = T \uplus U$.  
Because $>_\ackbo^\mul$ is closed under multiset sum, it suffices
to show that one of the cases~3(a,b,c) of
Definition~\ref{def:ackbo} holds for $S$ and $T$. Let $f = \rt(s)$ and
$g = \rt(t)$. We distinguish the following cases.
\begin{itemize}
\item
Suppose $f \nleqslant h$. We have
$S = \rrs[h]{S}{\nless} = \{ s \}$,
and from Lemmata~\ref{lem:subterm} and~\ref{lem:tf} we obtain
$S >_\ackbo^\mul T$. Since $T$ is a superset of
$\rrs[h]{T}{\nless} \uplus T{\restriction}_\VV - S{\restriction}_\VV$,
3(a) applies.
\item\smallskip
Suppose $f = h > g$. We have
$\rrs[h]{T}{\nless} \uplus T{\restriction}_\VV = \varnothing$.
If $\rrs[h]{S}{\nless} \neq \varnothing$, then
3(a) applies. Otherwise,
since AC symbols are binary and $T = \{ t \}$,
$|S| \geqslant 2 > 1 = |T|$. Hence 3(b) applies.
\item\smallskip
If $f = g = h$ then $s >_\ackbo t$ must be derived by one of the
cases~3(a,b,c) for $S$ and $T$.
\item\smallskip
Suppose $f, g < h$. We have 
$\rrs[h]{S}{\nless} = \rrs[h]{T}{\nless} \uplus T{\restriction}_\VV =
\varnothing$, $|S| = |T| = 1$, and
$\rrs[h]{S}{<} = \{ s \} >_\ackbo^\mul \{ t \} = \rrs[h]{T}{<}$.
Hence~3(c) holds.
\end{itemize}
Note that $f \geqslant g$ since $w(s) = w(t)$ and $s >_\ackbo t$.
Moreover, if $t \in \VV$ then $s = f^k(t)$ for some $k > 0$ with
$w(f) = 0$, which entails $f > h$ due to
the admissibility assumption.
\end{proof}

Closure under substitutions is the trickiest part since
by substituting AC-rooted terms for variables that appear in the
top\hyp flattening of a term, the structure of the term changes.
In the proof, the multisets $\{ t \in T \mid t \notin \VV \}$,
$\{ t\sigma \mid t \in T \}$, and $\{ \tf{f}(t) \mid t \in T \}$ are
denoted by $T{\restriction}_\FF$, $T\sigma$, and $\tf{f}(T)$,
respectively.

\begin{lemma}
\label{lem:f-key}
Let $>$ be a precedence, $f \in \FF_\AC$, and $({\GS},{\GT})$ an order pair
on terms such that $\GS$ and $\GT$ are closed under substitutions and
$f(x,y) \GT x, y$.
Consider terms $s$ and $t$ such that $S = \tf{f}(s)$,
$T = \tf{f}(t)$, $S' = \tf{f}(s\sigma)$, and $T' = \tf{f}(t\sigma)$.
\begin{enumerate}
\item
If $S \GT^f T$ then $S' \GT^f T'$.
\item
If $S \GS^f T$ then $S' \GT^f T'$ or $S' \GS^f T'$. In the
latter case
$|S| - |T| \leqslant |S'| - |T'|$ and
$\rrs{S'}{<} \GT^\mul \rrs{T'}{<}$ whenever
$\rrs{S}{<} \GT^\mul \rrs{T}{<}$.
\end{enumerate}
\end{lemma}

\begin{proof*}
Let $v$ be an arbitrary term. By the assumption on $\GT$ we have
either $\{ v \} = \tf{f}(v)$ or both $\{ v \} \GT^\mul \tf{f}(v)$ and 
$1 < |\tf{f}(v)|$. Hence, for any set $V$ of terms,
either $V = \tf{f}(V)$ or both $V \GT^\mul \tf{f}(V)$ and
$|V| < |\tf{f}(V)|$. Moreover, for $V = \tf{f}(v)$,
the following equalities hold:
\begin{xalignat*}{2}
\rrs{\tf{f}(v\sigma)}{\nless} & = 
\rrs{V}{\nless}\sigma \uplus
\rrs{\tf{f}(V{\restriction}_\VV\sigma)}{\nless}
&
\tf{f}(v\sigma){\restriction}_\VV & = 
\tf{f}(V{\restriction}_\VV\sigma){\restriction}_\VV
\end{xalignat*}
To prove the lemma, assume $S \REL^f T$
for ${\REL} \in \{ {\GS}, {\GT} \}$. We have
\(
\rrs{S}{\nless} \REL^\mul \rrs{T}{\nless} \uplus U
\)
where $U = (T - S){\restriction}_\VV$.
Since multiset extensions preserve closure under substitutions,
\(
\rrs{S}{\nless}\sigma\, \REL^\mul\, 
\rrs{T}{\nless}\sigma \uplus U\sigma
\)
follows. Using the above (in)equalities, we obtain
\begin{xalignat*}{2}
\rrs{S'}{\nless}
&=^{\phantom{\mul}}
\rrs{S}{\nless}\sigma
\uplus \rrs{\tf{f}(S{\restriction}_\VV\sigma)}{\nless}
\\
&\REL^\mul 
\rrs{T}{\nless}\sigma
\uplus \rrs{\tf{f}(S{\restriction}_\VV\sigma)}{\nless}
\uplus U\sigma
\\
&
\mathrel{O}^{\phantom{\mul}}
\rrs{T}{\nless}\sigma
\uplus \rrs{\tf{f}(S{\restriction}_\VV\sigma)}{\nless}
\uplus \tf{f}(U\sigma)
\\
&=^{\phantom{\mul}}
\rrs{T}{\nless}\sigma
\uplus \rrs{\tf{f}(S{\restriction}_\VV\sigma)}{\nless}
\uplus \tf{f}(U\sigma){\restriction}_\VV
\uplus \rrs{\tf{f}(U\sigma)}{\nless}
\uplus \rrs{\tf{f}(U\sigma)}{<}
\\
&\mathrel{P}^{\phantom{\mul}}
\rrs{T}{\nless}\sigma
\uplus \rrs{\tf{f}(T{\restriction}_\VV\sigma)}{\nless}
\uplus \tf{f}(U\sigma){\restriction}_\VV
\\
&=^{\phantom{\mul}}
\rrs{T}{\nless}\sigma
\uplus \rrs{\tf{f}(T{\restriction}_\VV\sigma)}{\nless}
\uplus \tf{f}(T{\restriction}_\VV\sigma){\restriction}_\VV
- \tf{f}(S{\restriction}_\VV\sigma){\restriction}_\VV
\\
&=^{\phantom{\mul}}
\rrs{T'}{\nless} \uplus T'{\restriction}_\VV - S'{\restriction}_\VV 
\end{xalignat*}
Here $O$ denotes $=$ if $U\sigma = \tf{f}(U\sigma)$ and $\GT^\mul$ if
$|U\sigma| < |\tf{f}(U\sigma)|$, while $P$ denotes $=$ if
$\rrs{U\sigma}{<} = \varnothing$ and $\supsetneq$ otherwise.
Since $({\GS^\mul},{\GT^\mul})$ is an order pair with
${\supseteq} \subseteq {\GS^\mul}$ and
${\supsetneq} \subseteq {\GT^\mul}$, we obtain $S' \REL^f T'$.

It remains to show 2. If $S' \NGT^f T'$ then $O$ and $P$ are both $=$
and thus $U\sigma = \tf{f}(U\sigma)$ and
$\rrs{U\sigma}{<} = \varnothing$.
Let $X = S{\restriction}_\VV \cap T{\restriction}_\VV$.
We have $U = T{\restriction}_\VV - X$.
\begin{itemize}
\item
Since $|W{\restriction}_\FF\sigma| = |W{\restriction}_\FF|$
and $|W| \leqslant |\tf{f}(W)|$ for an arbitrary set $W$ of terms,
we have $|S'| \geqslant |S| - |X| + |\tf{f}(X \sigma)|$.
From $|U\sigma| = |U| = |T{\restriction}_\VV| - |X|$ we obtain
\[
|T'| = |T{\restriction}_\FF \sigma| + |\tf{f}(U\sigma)| +
|\tf{f}(X \sigma)| = |T| - |X| + |\tf{f}(X \sigma)|
\]
Hence $|S| - |T| \leqslant |S'| - |T'|$ as desired.
\smallskip
\item
Suppose
$\rrs{S}{<} \GT^\mul \rrs{T}{<}$.
From $\rrs{U\sigma}{<} = \varnothing$ we infer
$\rrs{T{\restriction}_\VV\sigma}{<} \subseteq 
\rrs{S{\restriction}_\VV\sigma}{<}$.
Because 
\(
\rrs{S'}{<} = \rrs{S}{<}\sigma \uplus \rrs{S{\restriction}_\VV\sigma}{<}
\)
and 
\(
\rrs{T'}{<} = \rrs{T}{<}\sigma \uplus \rrs{T{\restriction}_\VV\sigma}{<}
\),
closure under substitutions of $\GT^\mul$ (which it inherits from 
$\GT$ and $\GS$) yields the desired
$\rrs{S'}{<} \GT^\mul \rrs{T'}{<}$.
\qed
\end{itemize}
\end{proof*}

\begin{lemma}
\label{lem:ackbo stable}
$>_\ackbo$ is closed under substitutions.
\end{lemma}
\begin{proof}
If $s >_\ackbo t$ is obtained by cases~0 or 1
in
Definition~\ref{def:ackbo}, the proof for standard KBO goes through.
If 3(a) or 3(b)
is used to obtain $s >_\ackbo t$,
according to Lemma~\ref{lem:f-key} one of these cases also applies to
$s\sigma >_\ackbo t\sigma$.
The final case is 3(c). So $\tf{f}(s){\restriction}_f^<
>_\ackbo^\mul \tf{f}(t){\restriction}_f^<$. Suppose 
$s\sigma >_\ackbo t\sigma$ cannot be obtained by
3(a) or 3(b). Lemma~\ref{lem:f-key}(2) yields
$|{\tf{f}(s\sigma)}| = |{\tf{f}(t\sigma)}|$ and
$\tf{f}(s\sigma){\restriction}_f^< >_\ackbo^\mul 
\tf{f}(t\sigma){\restriction}_f^<$.
Hence case~3(c) is applicable to obtain $s\sigma >_\ackbo t\sigma$.
\end{proof}

We arrive at the main theorem of this section.

\begin{theorem}
\label{thm:ackbo}
The order $>_\ackbo$ is an AC-compatible simplification order.
\qed
\end{theorem}

Since we deal with finite non-variadic signatures, simplification orders
are well-founded.
%
%
The following example shows that AC-KBO is not \emph{incremental}, i.e.,
orientability is not necessarily preserved when the precedence is
extended.
This is in contrast to the
AC-RPO of \citeN{R02}.
However, this is not necessarily a disadvantage;
actually, the example shows that by allowing partial precedences more
TRSs can be proved to be AC terminating using AC-KBO.


\begin{example}
\label{example partiality}
Consider the TRS $\RR$ consisting of the rules
\begin{xalignat*}{2}
\m{a} \circ (\m{b} \bullet \m{c})
&\to \m{b} \circ \m{f}(\m{a} \bullet \m{c}) &
\m{a} \bullet (\m{b} \circ \m{c})
&\to \m{b} \bullet \m{f}(\m{a} \circ \m{c})
\end{xalignat*}
over the signature
$\FF = \{ \m{a}, \m{b}, \m{c}, \m{f}, {\circ}, {\bullet} \}$
with ${\circ}, {\bullet} \in \FF_\AC$. By taking
the precedence $\m{f} > \m{a}, \m{b}, \m{c}, {\circ}, {\bullet}$
and admissible weight function $(w,w_0)$ with
\begin{align*}
w(\m{f}) &= w({\circ}) = w({\bullet}) = 0 &
w_0 &= w(\m{a}) = w(\m{c}) = 1 &
w(\m{b}) &= 2
\end{align*}
the resulting $>_\ackbo$ orients both rules from left to right.
It is essential that $\circ$ and $\bullet$ are incomparable in the
precedence:
We must have $w(\m{f}) = 0$, so
$\m{f} > \m{a}, \m{b}, \m{c}, {\circ}, {\bullet}$ is enforced by
admissibility. If ${\circ} > {\bullet}$ then the first rule can only
be oriented from left to right if
$\m{a} >_\ackbo \m{f}(\m{a} \bullet \m{c})$ holds, which contradicts
the subterm property. If ${\bullet} > {\circ}$ then we use the
second rule to obtain the impossible
$\m{a} >_\ackbo \m{f}(\m{a} \circ \m{c})$.
Similarly, $\RR$ is also orientable by $>_\KVC$ but we must
adopt a non-total precedence.
\end{example}

The easy proof of the final theorem in this section can be found
in the online appendix.

\begin{theorem}\label{thm:total}
If $>$ is total then $>_\ackbo$ is AC-total on ground terms.
\end{theorem}

\section{Complexity}
\label{complexity}

In this section we discuss complexity issues for the orders
defined in the preceding sections. We start with the membership
problem: Given two terms $s$ and $t$, a weight function, and a precedence,
does $s > t$ hold? For plain KBO this problem is known to be
decidable in linear time~\cite{L06b}. For
$>_{\steinbach}$, $>_\KV$, and $>_\ackbo$ we show the problem to be 
decidable in polynomial time, but we start with the unexpected result that
$>_\KVC$ membership is NP-complete. For NP-hardness we use the
reduction technique of \citeN[Theorem~4.2]{TAN12}.



\begin{theorem}
\label{thm:KVC NP-hard}
The decision problem for $>_\KVC$ is NP-complete.
\end{theorem}
\begin{proof}
We start with NP-hardness.
It is sufficient to show NP-hardness of deciding
$S >_\kvc^{\mul} T$
since we can easily construct terms $s$ and $t$
such that $S >_\kvc^{\mul} T$ if and only if $s >_\KVC t$.
To wit, for $S = \{ \seq{s} \}$ and $T = \{ \seq[m]{t} \}$
we introduce an AC symbol $\circ$ and constants $\m{c}$ and
$\m{d}$ such that $\circ > \m{c}, \m{d}$ and define
\begin{align*}
s &= s_1 \circ \dots \circ s_n \circ \m{c} &
t &= t_1 \circ \dots \circ t_m \circ \m{d} \circ \m{d}
\end{align*}
The weights of $\m{c}$ and $\m{d}$ should be chosen so that $w(s) = w(t)$.
If $S >_\kvc^\mul T$ then case~3(a) applies for $s >_\KVC t$.
Otherwise, $S \geqslant_\kvc^\mul T$ implies $n = m$ and thus
$|\tf{\circ}(s)| < |\tf{\circ}(t)|$.
Hence neither case~3(b) nor 3(c) applies.

We reduce a non-empty CNF SAT problem $\phi = \{ C_1, \dots, C_m \}$ over
propositional variables $\seq{x}$ to the decision problem
$S_\phi >_\kvc^\mul T_\phi$. The multisets
$S_\phi$ and $T_\phi$ will consist of terms in
$\TT(\{ \m{a}, \m{f} \},\{ \seq{x},\seq[m]{y} \} )$,
where $\m{a}$ is a constant with $w(\m{a}) = w_0$ and $\m{f}$ has arity
$m+1$. For each $1 \leqslant j \leqslant m$ and literal $l$, we define
\[
s_j(l) = \begin{cases}
y_j & \text{if $l \in C_j$} \\
\m{a} & \text{otherwise}
\end{cases}
\]
Moreover, for each $1 \leqslant i \leqslant n$ we define
\begin{align*}
t_i^+ &= \m{f}(x_i, s_1(x_i), \dots, s_m(x_i)) &
t_i^- &= \m{f}(x_i, s_1(\neg x_i), \dots, s_m(\neg x_i))
\end{align*}
and $t_i = \m{f}(x_i, \m{a}, \dots, \m{a})$.
Note that $w(t_i^+) = w(t_i^-) = w(t_i) > w(y_j)$ for all
$1 \leqslant i \leqslant n$ and $1 \leqslant j \leqslant m$.
Finally, we define
\begin{xalignat*}{2}
S_\phi &= \{ t_1^+, t_1^-, \dots, t_n^+, t_n^- \} &
T_\phi &= \{ \seq{t}, \seq[m]{y} \}
\end{xalignat*}
Note that for every $1 \leqslant i \leqslant n$ there is no
$s \in S_\phi$ such that $s >_\Wt t_i$.
Hence $S_\phi >_\kvc^\mul T_\phi$ if and only if
$S_\phi$ can be written as $\{ \seq{s}, \seq{s'} \}$
such that $s_i \geqslant_\kvc t_i$ for all
$1 \leqslant i \leqslant n$,
and for all $1 \leqslant j \leqslant m$ there exists
an $1 \leqslant i \leqslant n$ such that $s'_i >_\Wt y_j$.
It is easy to see that the only candidates for $s_i$ are $t_i^+$ and
$t_i^-$.

Now suppose $S_\phi >_\kvc^\mul T_\phi$
with $S_\phi$ written as above.
Consider the assignment $\alpha$ defined as follows:
$\alpha(x_i)$ is true if and only if $s_i = t_i^-$. We claim
that $\alpha$ satisfies every $C_j \in \phi$.
We know that there exists $1 \leqslant i \leqslant n$ such that
$s'_i >_\Wt y_j$ and thus also
$y_j \in \Var(s'_i)$. This is only
possible if $x_i \in C_j$ (when $s'_i = t_i^+$)
or $\neg x_i \in C_j$ (when $s'_i = t_i^-$). Hence, by construction
of $\alpha$, $\alpha$ satisfies $C_j$.

Conversely, suppose $\alpha$ satisfies $\phi$. Let $s'_i = t_i^+$ and
$s_i = t_i^-$ if $\alpha(x_i)$ is true and $s'_i = t_i^-$ and
$s_i = t_i^+$ if $\alpha(x_i)$ is false. We trivially have
$s_i \geqslant_\kvc t_i$ for all
$1 \leqslant i \leqslant n$.
Moreover, for each $1 \leqslant j \leqslant m$, $C_j$ contains a literal
$l = (\neg) x_i$ such that $\alpha(l)$ is true.
By construction, $y_j \in \Var(s'_i)$ and thus $s'_i >_\Wt y_j$.
Since $\phi$ is non-empty, $m > 0$ and hence $S_\phi >_\kvc^\mul T_\phi$
as desired.

To obtain NP-completeness we need to show membership in NP, which is
easy; one just guesses how the terms in
the various multisets relate to each other in order to satisfy
the multiset comparisons in the definition of $>_\KVC$.
\end{proof}

Next we show that the complexity of deciding $>_\KV$
and $>_\ackbo$ for given weights and precedence is
decidable in polynomial time.
Given a sequence $S = \seq{s}$ and an index $1 \leqslant i \leqslant n$, we
denote by $S[t]_i$ the sequence
obtained by replacing $s_i$ with $t$ in $S$, and
by $S[\,]_i$ the sequence obtained by removing $s_i$ from $S$.
Moreover, we write $\{ S \}$ as a shorthand for the multiset 
$\{ \seq{s} \}$.

\begin{lemma}
\label{lem:equivalence}
Let $(\GS,\GT)$ be an order pair 
such that
${\sim} := {\GS} \setminus {\GT}$ is symmetric. 
If $s \sim t$ then
$M \mathrel{\uplus} \{ s \} \GT^\mul N \mathrel{\uplus} \{ t \}$
and $M \GT^\mul N$ are equivalent.
\end{lemma}
\begin{proof*}
We only show that
$M \mathrel{\uplus} \{ s \} \GT^\mul N \mathrel{\uplus} \{ t \}$
implies $M \GT^\mul N$, since the other direction is trivial.
So suppose
$M \uplus \{ s \} \sqsupset_k^\mul N \uplus \{ t \}$,
where sequences $S = \seq[m]{s}$ and $T = \seq{t}$ 
satisfy the conditions for $\sqsupset_k^\mul$
in Definition~\ref{lex and mul}.
Because we have $\{S\} = M \uplus \{s\}$ and
$\{T\} = N \uplus \{t\}$, there are indices $i$ and $j$ such that
$s = s_i$ and $t = t_j$.
In order to establish $M \GT^\mul N$ we distinguish four cases.
\begin{itemize}
\item
If $i, j \leqslant k$ then
$s_j \GS t_j = t \sim s = s_i \GS t_i$ and thus
\( \{ S[s_j]_i[\,]_j \} \sqsupset_{k-1}^\mul 
\{ T[\,]_j \}
\).
\smallskip
\item
If $i \leqslant k < j$ then
there exists some $l > k$ such that
$s_l \GT t_j = t \sim s = s_i \GS t_i$.
Therefore,
$\{ S[\,]_i \} \sqsupset_{k-1}^\mul \{ T[t_i]_j[\,]_i \}$.
\smallskip
\item
If $j \leqslant k < i$ then
$s_j \GS t_j = t \sim s = s_i$ and thus
$s_j \GT t_l$ for every $l > k$ such that $s_i \GT t_l$.
Hence $\{ S[s_j]_i[\,]_j \} \sqsupset_{k-1}^\mul \{ T[\,]_j \}$.
\smallskip
\item
The remaining case $k < i, j$ is analogous to the previous
case, and we obtain
$\{ S[\,]_i \} \sqsupset_k^\mul \{ T[\,]_j \}$.
\end{itemize}
Because $\{ S[s_j]_i[\,]_j \} = \{ S[\,]_i \} = M$ and
$\{ T[t_i]_j[\,]_i \} = \{ T[\,]_j \} = N$ hold,
in all cases $M \GT^\mul N$ is concluded.
\qed
\end{proof*}

\begin{lemma}
\label{thm:P}
Let $(\GS,\GT)$ be an order pair such that 
${\sim} := {\GS} \setminus {\GT}$ is symmetric
and the decision problems for $\GS$ and $\GT$ are in P.
Then the decision problem for $\GT^\mul$ is in P.
\end{lemma}
\begin{proof}
Suppose we want to decide whether two multisets $S$ and $T$ satisfy
$S \GT^\mul T$.
We first check if there exists a pair $(s,t) \in S \times T$ such that
$s \sim t$, which can be done by testing $s \GS t$ and $s \NGT t$ at most
$|S| \times |T|$ times. If such a pair is found then according to
Lemma~\ref{lem:equivalence}, the problem is reduced to
$S - \{ s \} \GT^\mul T - \{ t \}$.  
Otherwise, we check for each $t \in T$ whether there exists $s \in S$ such
that $s \GT t$, which can be done by testing $s \GT t$ at most
$|S| \times |T|$ times.
\end{proof}

Using the above lemma, we obtain the following result by a straightforward
induction argument.

\begin{corollary}
The decision problems for $>_\ackbo$, $>_\KV$, and $>_\steinbach$
belong to P.
\qed
\end{corollary}

Next we address the complexity of the important orientability problem:
Given a TRS $\RR$, do there exist a weight function and a precedence such
that the rules of $\RR$ are oriented from left to right with respect
to the order under consideration?
It is well-known~\cite{KV03} that KBO orientability is decidable in
polynomial time. We show that $>_\KV$ and $>_\ackbo$
orientability are
NP-complete even for ground TRSs.
First we show NP-hardness of $>_\KV$ orientability by a
reduction from SAT.

Let $\phi = \{ \seq{C} \}$ be a CNF SAT problem over propositional
variables $\seq[m]{p}$. We consider the signature $\FF_\phi$ consisting of
an AC symbol $+$, constants $\Bot$ and $\seq{\Top}$, and unary function
symbols $\seq[m]{p}$, $\High$, $\Low$, and $\ToP_i^j$ for all
$i \in \{ 1, \dots, n \}$ and $j \in \{ 0, \dots, m \}$.  We define a
ground TRS $\RR_\phi$ on $\TT(\FF_\phi)$ such that $>_\KV$ orients
$\RR_\phi$ if and only if $\phi$ is satisfiable.  The TRS $\RR_\phi$ will
contain the following base system $\RR_0$ that enforces certain constraints
on the precedence and the weight function:
\begin{gather*}
\High(\Bot + \Bot) \to \High(\Bot) + \Bot \qquad
\Low(\Bot) + \Bot \to \Low(\Bot + \Bot) \qquad
\High(\Low(\Low(\Bot))) \to \Low(\High(\High(\Bot))) \\
\High(p_1(\Bot)) \to \Low(p_2(\Bot)) \qquad
\cdots \qquad
\High(p_m(\Bot)) \to \Low(\High(\Bot)) \qquad
\High(\High(\Bot)) \to \Low(p_1(\Bot))
\end{gather*}

\begin{lemma}
\label{lem:base}
The order $>_\KV$ is compatible with $\RR_0$ if and only if
$\High > + > \Low$ and $w(\High) = w(\Low) = w(p_j)$ for all
$1 \leqslant j \leqslant m$.
\qed
\end{lemma}

Consider the clause $C_i$ of the form
$\{ \seq[k]{p'}, \seq[l]{\neg p''} \}$. Let
$U$, $U'$, $V$, and $W$ denote the following multisets:
\label{encoding}
\begin{xalignat*}{2}
U &= \{ p_1'(\Low(\Top_i)), \dots, p_k'(\Low(\Top_i)) \} &
V &= \{ p_0''(\ToP_i^{0,1}), \dots, p_{l-1}''(\ToP_i^{l-1,l}),
p_l''(\ToP_i^{l,0}) \} \\
U' &= \{ \Low(p_1'(\Top_i)), \dots, \Low(p_k'(\Top_i)) \} &
W  &= \{ p_0''(\ToP_i^{0,0}), \dots, p_l''(\ToP_i^{l,l}) \}
\end{xalignat*}
where we write $p''_0$ for $\High$ and $\ToP_i^{j,k}$ for
$\ToP_i^j(\ToP_i^k(\Bot))$.  
The TRS $\RR_\phi$ is defined as
the union of $\RR_0$ and
$\{ \ell_i \to r_i \mid 1 \leqslant i \leqslant n \}$
with
\begin{xalignat*}{2}
\ell_i & = \Low(\Low(\Bot + \Bot)) +
\textstyle \sum U + \sum V
&
r_i & = \Low(\Bot) + \Low(\Bot) +
\textstyle \sum U' + \sum W
\end{xalignat*}
Note that the symbols $\Top_i$ and $\ToP_i^0, \dots, \ToP_i^l$
are specific to the rule $\ell_i \to r_i$.

\begin{example}
Consider a clause $C_1 = \{ x, \neg y, \neg z \}$.
We have
\begin{align*}
\ell_1 &=
  \Low(\Low(\Bot + \Bot)) +
  x(\Low(\Top_i)) +
  \High(\ToP_1^0(\ToP_1^1(\Bot))) +
  y(\ToP_1^1(\ToP_1^2(\Bot))) +
  z(\ToP_1^2(\ToP_1^0(\Bot))) \\
r_1 &=
  \Low(\Bot) + \Low(\Bot) +
  \Low(x(\Top_i)) +
  \High(\ToP_1^0(\ToP_1^0(\Bot))) +
  y(\ToP_1^1(\ToP_1^1(\Bot))) +
  z(\ToP_1^2(\ToP_1^2(\Bot)))
\end{align*}
Note that $x$, $y$, and $z$ are unary function symbols. We have
$w(\ell_1) = w(r_1)$ for any weight function $w$.
Suppose $\High > + > \Low$ and $w(\High) = w(\Low) = w(x) = w(y) = w(z)$.

We consider a number of cases, depending on the order of
$x$, $y$, $z$, and $+$ in the precedence.
If $x, y, z > +$ (i.e., $x$, $y$, and $z$ are assigned true)
then $\ell_1 >_\KV r_1$ can be satisfied by choosing
$w(\Top_1)$ large enough such that
$w(x(\Low(\Top_1))) > w(t)$ for all $t \in \rrs[+]{\tf{+}(r_1)}{>}$, where
\begin{align*}
\rrs[+]{\tf{+}(\ell_1)}{>} &~=~ \{ x(\Low(\Top_1)),
 \High(\ToP_1^0(\ToP_1^1(\Bot))),
  y(\ToP_1^1(\ToP_1^2(\Bot))),
  z(\ToP_1^2(\ToP_1^0(\Bot))) \} \\
\rrs[+]{\tf{+}(r_1)}{>} &~=~ \{ \phantom{x(\Low(\Top_1)),{}}
 \High(\ToP_1^0(\ToP_1^0(\Bot))),
  y(\ToP_1^1(\ToP_1^1(\Bot))),
  z(\ToP_1^2(\ToP_1^2(\Bot))) \}
\intertext{%
On the other hand, if $y, z > + > x$ (i.e., $x$ is falsified)
then $\ell_1 >_\KV r_1$ is not satisfiable;
no matter how we assign weights to $\ToP_1^0$, $\ToP_1^1$, and $\ToP_1^2$,
a term in $\tf{+}(r_1)$ has the maximum weight, where
}
\rrs[+]{\tf{+}(\ell_1)}{>} &~=~ \{
  \High(\ToP_1^0(\ToP_1^1(\Bot))),
  y(\ToP_1^1(\ToP_1^2(\Bot))),
  z(\ToP_1^2(\ToP_1^0(\Bot))) \} \\
\rrs[+]{\tf{+}(r_1)}{>} &~=~ \{
  \High(\ToP_1^0(\ToP_1^0(\Bot))),
  y(\ToP_1^1(\ToP_1^1(\Bot))),
  z(\ToP_1^2(\ToP_1^2(\Bot))) \}
\intertext{%
However, if $y > + > x, z$ (i.e.\ $z$ is falsified) then
$\ell_1 >_\KV r_1$ can be satisfied by choosing $w(\ToP_1^2)$ large
enough, where
}
\rrs[+]{\tf{+}(\ell_1)}{>} &~=~ \{
  \High(\ToP_1^0(\ToP_1^1(\Bot))),
  y(\ToP_1^1(\ToP_1^2(\Bot))) \} \\
\rrs[+]{\tf{+}(r_1)}{>} &~=~ \{
  \High(\ToP_1^0(\ToP_1^0(\Bot))),
  y(\ToP_1^1(\ToP_1^1(\Bot))) \}
\intertext{%
Similarly, if $+ > x, y, z$ then $\ell_1 >_\KV r_1$ can be satisfied by
choosing $w(\ToP_1^1)$ large enough, where
}
\rrs[+]{\tf{+}(\ell_1)}{>} &~=~ \{ \High(\ToP_1^0(\ToP_1^1(\Bot))) \} \\
\rrs[+]{\tf{+}(r_1)}{>} &~=~ \{ \High(\ToP_1^0(\ToP_1^0(\Bot))) \}
\end{align*}
\end{example}

\begin{lemma}
\label{lem:clause}
Let $\High > + > \Low$. Then, $\RR_\phi \subseteq {>_\KV}$
for some $(w, w_0)$ if and only if for every $i$
there is some $p$ such that $p \in C_i$ with $p \nless +$ or
$\neg p \in C_i$ with $+ > p$.
\end{lemma}
\begin{proof}
For the ``if'' direction we reason as follows.
Consider a (partial) weight function $w$ such that
$w(\High) = w(\Low) = w(p_j)$ for all $1 \leqslant j \leqslant m$.
We obtain $\RR_0 \subseteq {>_\KV}$ from Lemma~\ref{lem:base}.
Furthermore, consider 
$C_i = \{ p'_1,\dots,p'_k, \neg p''_1,\dots, \neg p''_l \}$ and
$\ell_i$, $r_i$, $U$, $V$ and $W$ defined above.
Let $L = \tf{+}(\ell_i)$ and $R = \tf{+}(r_i)$.
We clearly have
$\rrs[+]{L}{\nless} = \rrs[+]{U}{\nless} \cup \rrs[+]{V}{\nless}$ and
$\rrs[+]{R}{\nless} = \rrs[+]{W}{\nless}$.
It is easy to show that $w(\ell_i) = w(r_i)$.
We show $\ell_i >_\KV r_i$ by distinguishing two cases.
\begin{enumerate}
\item
First suppose that
$p'_j \nless +$ for some $1 \leqslant j \leqslant k$.
We have $p'_j(\Low(\Top_i)) \in \rrs[+]{U}{\nless}$.
Extend the weight function $w$ such that
\[
w(\Top_i) = 1 + 2 \cdot \max\,\{ w(\ToP_i^0), \dots, w(\ToP_i^l) \}
\]
Then $p'_j(\Low(\Top_i)) >_\Wt t$ for all terms $t \in W$ and hence
$\rrs[+]{L}{\nless} >_\Wt^\mul \rrs[+]{R}{\nless}$.
Therefore $\ell_i >_\KV r_i$ by case~3(a).
\smallskip
\item
Otherwise, $\rrs[+]{U}{\nless} = \varnothing$ holds. By assumption
$+ > p''_j$ for some $1 \leqslant j \leqslant l$.
Consider the smallest $m$ such that $+ > p''_m$.
Extend the weight function $w$ such that
\[
w(\ToP_i^m) = 1 + 2 \cdot \max\,\{ w(\ToP_i^j) \mid j \neq m \}
\]
Then
$w(p''_{m-1}(\ToP_i^{m-1,m})) > w(p''_j(\ToP_i^{j,j}))$
for all $j \neq m$. From $p''_{m-1} > +$ we infer
$p''_{m-1}(\ToP_i^{m-1,m}) \in \rrs[+]{V}{\nless}$.
(Note that $p_{m-1}'' = \High > +$ if $m = 1$.)
By definition of $m$,
$p''_m(\ToP_i^{m,m}) \notin \rrs[+]{W}{\nless}$.
It follows that
$\rrs[+]{L}{\nless} >_\Wt^\mul \rrs[+]{R}{\nless}$ and thus
$\ell_i >_\KV r_i$ by case~3(a).
\end{enumerate}
Next we prove the ``only if'' direction. So suppose
there exists a weight function $w$ such that
$\RR_\phi \subseteq {>_\KV}$.
We obtain $w(\High) = w(\Low) = w(p_j)$ for all $1 \leqslant j \leqslant m$
from Lemma~\ref{lem:base}. It follows that $w(\ell_i) = w(r_i)$
for every $C_i \in \phi$.
Suppose for a proof by contradiction that
there exists $C_i \in \phi$ such that
$+ > p$ for all $p \in C_i$ and
$p \nless +$ whenever $\neg p \in C_i$.
So $\rrs[+]{L}{\nless} = V$ and $\rrs[+]{R}{\nless} = W$.
Since $|R| = |L| + 1$, we must have $\ell_i >_\KV r_i$ by
case~3(a) and thus $V >_\Wt W$.
Let $s$ be a term in $V$ of maximal weight.
We must have $w(s) \geqslant w(t)$ for all terms $t \in W$. By
construction of the terms in $V$ and $W$, this is only possible if all
symbols $\ToP_i^j$ have the same weight. It follows that all terms
in $V$ and $W$ have the same weight. Since $|V| = |W|$ and
for every term $s' \in V$ there exists a unique term
$t' \in W$ with $\rt(s') = \rt(t')$, we conclude
$V =_\Wt W$, which provides the desired contradiction.
\end{proof}

After these preliminaries we are ready to prove NP-hardness.

\begin{theorem}
\label{KV orientation NP-hard}
The (ground) orientability problem for $>_\KV$ is NP-hard.
\end{theorem}
\begin{proof}
It is sufficient to prove that a CNF formula $\phi = \{ \seq{C} \}$ is
satisfiable if and only if the corresponding $\RR_\phi$ is orientable by
$>_\KV$.  Note that the size of $\RR_\phi$ is linear in the size of $\phi$.
First suppose that $\phi$ is satisfiable. Let $\alpha$ be a satisfying
assignment for the atoms $\seq[m]{p}$. Define the precedence $>$ as
follows: $\High > + > \Low$ and $p_j > +$ if $\alpha(p_j)$ is true and
$+ > p_j$ if $\alpha(p_j)$ is false.
Then $\RR_\phi \subseteq {>_\KV}$ follows from Lemma~\ref{lem:clause}.
Conversely, if $\RR_\phi$ is compatible with $>_\KV$ then we define
an assignment $\alpha$ for the atoms in $\phi$ as follows:
$\alpha(p)$ is true if $p \nless +$ and
$\alpha(p)$ is false if $+ > p$.
We claim that
$\alpha$ satisfies $\phi$. Let $C_i$ be a clause in $\phi$.
According to Lemma~\ref{lem:clause}, $p \nless +$ for one of
the atoms $p$ in $C_i$ or $+ > p$ for one of the negative
literals $\neg p$ in $C_i$. Hence $\alpha$ satisfies $C_i$ by
definition.
\end{proof}

We can show NP-hardness of $>_\ackbo$ by adapting the above
construction accordingly, as shown in
Appendix~\ref{sec:ACKBO NP-hard}.

\begin{theorem}
\label{ACKBO orientation NP-hard}
The (ground) orientability problem for $>_\ackbo$ is NP-hard.
\qed
\end{theorem}

The NP-hardness results of Theorems~\ref{KV orientation NP-hard}
and~\ref{ACKBO orientation NP-hard} can be strengthened to
NP-completeness. This is not entirely trivial because there are
infinitely many different weight functions to consider.

\begin{lemma}
\label{ACKBO in NP}
The orientability problems for $>_\ackbo$ and $>_\KV$ belong to NP.
\end{lemma}

\begin{proof}[Proof (sketch)]
We sketch the proof for $>_\ackbo$.
With minor modifications the result for $>_\KV$ is obtained.
\\\indent
For each rule $\ell \to r$ of a given TRS $\RR$ we guess which choices are
made in the definition of $>_\ackbo$ when evaluating $\ell >_\ackbo r$.
In particular, we do not guess the weight function, but rather the
comparison ($=$ or $>$) of the weights of certain subterms of $\ell$ and
$r$. These comparisons are transformed into constraints on the weight
function by symbolically evaluating the weight expressions.
We add the constraints stemming from the definition of the weight 
function. The resulting
problem is a conjunction of linear constraints over unknowns (the weights
of the function symbols and $w_0$) over the integers.
It is well-known \cite[Section~10.3]{S86} that solving such a
\emph{linear program} over the rationals can be done in polynomial time.
If there is a solution we check
%
%
%
%
the admissibility condition and well-foundedness of the precedence.
(If an integer valued weight function is desired, one can simply multiply
the weights by the least common multiple of their denominators. This
induces the same weight order on terms and does not affect the
admissibility condition.)
\\\indent
Since there are polynomially (in the size of the compared terms) many
choices in the definition of $>_\ackbo$
and each choice can be checked for correctness in polynomial time,
membership in NP follows.
\end{proof}

\begin{corollary}
The orientability problems for $>_\ackbo$ and $>_\KV$ are NP-complete.
\qed
\end{corollary}

The NP-hardness proofs of $>_\KV$ and $>_\ackbo$ orientability given
earlier do not extend to $>_\steinbach$ since the latter requires that
AC symbols are minimal in the precedence.

We conjecture that the orientability problem for $>_\steinbach$
belongs to P.

\section{AC-RPO}
\label{AC-RPO}

In this section we compare AC-KBO with AC-RPO~\cite{R02}.
Since the latter is incremental \cite[Lemma~22]{R02},
we restrict the discussion to total precedences.

\begin{definition}
Let $>$ be a precedence and $t = f(u,v)$ such that $f \in \FF_\AC$ and 
$\tf{f}(t) = \{ \seq{t} \}$. We write
$t \embsm{f} u$ for all terms $u$ such that 
$\tf{f}(u) = \{ t_1,\ldots,t_{i-1},s_j,t_{i+1},\ldots,t_n \}$ for some 
$t_i = g(\seq[m]{s})$ with $f > g$ and $1 \leqslant j \leqslant m$.
\end{definition}

Using previously introduced notations, AC-RPO can be defined as follows.

\begin{definition}
\label{def:acrpo}
Let $>$ be a precedence
and let $\FF \setminus \FF_\AC = \FF_{\mul} \uplus \FF_{\lex}$. 
We define $>_\acrpo$ inductively as follows:
$s >_\acrpo t$ if one of the following conditions holds:
\begin{enumerate}
\setcounter{enumi}{-1}
\item
$s = f(\seq s)$ and $s_i \geqslant_\acrpo t$ for some 
$1 \leqslant i \leqslant n$,
\smallskip
\item
$s = f(\seq{s})$, $t = g(\seq[m]{t})$, $f > g$, and $s >_\acrpo t_j$ for 
all $1 \leqslant j \leqslant m$,
\smallskip
\item
$s = f(\seq{s})$, $t = f(\seq{t})$, $f \notin \FF_\AC$,
$s >_\acrpo t_j$ for all $1 \leqslant j \leqslant n$, and
either
\begin{itemize}
\item[(a)]
$f \in \FF_{\textsf{lex}}$ and $(\seq{s}) >_\acrpo^\lex (\seq{t})$, or
\smallskip
\item[(b)]
$f \in \FF_{\textsf{mul}}$ and
$\{ \seq{s} \} >_\acrpo^\mul \{ \seq{t} \}$,
\end{itemize}
\item
$s = f(s_1,s_2)$, $t = f(t_1,t_2)$, $f \in \FF_\AC$, and 
$s' \geqslant_\acrpo t$ for some $s'$ such that $s \embsm{f} s'$,
\smallskip
\item
$s = f(s_1,s_2)$, $t = f(t_1,t_2)$, $f \in \FF_\AC$, 
$s >_\acrpo t'$ for all $t'$ such that $t \embsm{f} t'$,
and for $S = \tf{f}(s)$ and $T = \tf{f}(t)$
\begin{itemize}
\item[(a)]
$S >_\acrpo^f T$,
\smallskip
\item[(b)]
$S =_\AC^f T$ and $|S| > |T|$, or
\smallskip
\item[(c)]
$S =_\AC^f T$, $|S| = |T|$, and
$\rrs{S}{<} >_\acrpo^{\mul} \rrs{T}{<}$.
\end{itemize}
\end{enumerate}
The relation $=_\AC$ is used as preorder in
$>_\acrpo^\lex$ and $>_\acrpo^\mul$, and as equivalence relation in 
$\geqslant_\acrpo$.
\end{definition}

\begin{example}
Consider the TRS $\RR$ consisting of the rules
\begin{xalignat*}{3}
\m{f}(x) + \m{g}(x) &\to \m{g}(x) + (\m{g}(x) + \m{g}(x)) &
\m{f}(x) &\to \m{g}(x) + \m{a}
\end{xalignat*}
over the signature $\FF = \{ \m{f}, \m{g}, +, \m{a} \}$
with $+ \in \FF_\AC$. Let $\RR'$ be the TRS obtained from $\RR$
by reverting the first rule.
When using AC-RPO with precedence $\m{f} > + > \m{g} > \m{a}$,
both rules in $\RR$ can be oriented from left to right. Since the 
second rule requires $\m{f} > +$ and $\m{f} > \m{g}$, termination of
$\RR'$ cannot be shown with AC-RPO.

In contrast, AC-KBO cannot orient $\RR$ due to the variable condition.
But the precedence $\m{g} > + > \m{f} > \m{a}$ 
and admissible weight function $(w,w_0)$ with
$w(+) = 0$, $w_0 = w(\m{g}) = w(\m{a}) = 1$ and $w(\m{f}) = 3$ allows
the resulting $>_\ackbo$ to orient both rules of $\RR'$.
\end{example}

Case~4 in Definition~\ref{def:acrpo} differs from the original version 
in~\cite{R02} in that we used notions introduced for AC-KBO.
We now recall the original definition and prove the two versions
equivalent in Lemma~\ref{lem:acrpo}.

\begin{definition}
For $S = \{ \seq{s} \}$ let $\#(S) = \#(s_1) + \cdots + \#(s_n)$ where
$\#(s_i) = s_i$ for $s_i \in \VV$ and $\#(s_i) = 1$ otherwise. Then 
$\#(S) > \#(T)$ ($\#(S) \geqslant \#(T)$) is defined via comparison of 
linear polynomials over the positive integers.

Let $>$ be a total precedence.
The order $>_\acrpoo$ is inductively defined as in
Definition~\ref{def:acrpo}, but with case~4 as follows:
\begin{itemize}
\item[4$'$.]
$s = f(s_1,s_2)$, $t = f(t_1,t_2)$, $f \in \FF_\AC$,
$s >_{\acrpoo} t'$ for all $t'$ such that $t \embsm{f} t'$,
$\rrs{S}{>} \uplus S{\restriction}_\VV \geqslant_{\acrpoo}^\mul 
\rrs{T}{>} \uplus T{\restriction}_\VV$ for $S = \tf{f}(s)$ and
$T = \tf{f}(t)$, and
\begin{itemize}
\item[(a)]
\smallskip
$\rrs{S}{>} >_{\acrpoo}^\mul \rrs{T}{>}$, or
\item[(b)]
\smallskip
$\#(S) > \#(T)$, or
\item[(c)]
\smallskip
$\#(S) \geqslant \#(T)$, and $S >_{\acrpoo}^\mul T$.
\end{itemize}
\end{itemize}
\end{definition}

The proof of the following correspondence can be found in the online appendix.
\begin{lemma}
\label{lem:acrpo}
Let $>$ be a total precedence.
We have $s >_\acrpo t$ if and only if $s >_\acrpoo t$. 
\end{lemma}

It is known that both orientability and membership are NP-hard for the
multiset path order~\cite{KN85}. It is not hard to adapt these proofs to
LPO, and NP-hardness for the case of RPO is an easy consequence.

In contrast to AC-KBO, a straightforward application of the definition 
of AC-RPO (in particular case~4 of Definition~\ref{def:acrpo}) may 
generate an exponential number of subproblems, as illustrated by the 
following example.

\begin{example}
Consider the signature
$\FF = \{ \m{f}, \m{g}, \m{h}, \circ \}$ with $\circ \in \FF_\AC$ and
precedence $\m{f} > {\circ} > \m{g} > \m{h}$.
Let $t = x \circ y$ and $t_n = t\sigma^n$ for the substitution
$\sigma = \{ x \mapsto \m{g}(x) \circ \m{h}(y), y \mapsto \m{h}(y) \}$.
The size of $t_n$ is quadratic in $n$ but the number of terms $u$
that satisfy $t_n \mathrel{(\embsm{\circ})^+} u$ is exponential 
in $n$. Now suppose one wants to decide whether 
$\m{f}(x) \circ \m{f}(y) >_\acrpo t_n$ holds.
Only case~4(a) is applicable but in order to conclude orientability,
case~4(a) needs to be applied recursively in order to verify 
$\m{f}(x) \circ \m{f}(x) >_\acrpo  u$ for the exponentially many 
terms $u$ such that $t_n \mathrel{(\embsm{\circ})^+} u$.
\end{example}

\section{Subterm Coefficients}
\label{subterm coefficients}

Subterm coefficients were introduced in \cite{LW07} in order to
cope with rewrite rules like $\m{f}(x) \to \m{g}(x,x)$ which violate
the variable condition.
A \emph{subterm coefficient function} is a partial mapping
$\scoeff \colon \FF \times \Nat \to \Nat$ 
such that for a function symbol $f$ of arity $n$ we have
$\scoeff(f,i) > 0$ for all $1 \leqslant i \leqslant n$.
Given a weight function $(w,w_0)$ and a subterm coefficient 
function $\scoeff$, the weight of a term is inductively defined as follows:
\[
w(t) =
\begin{cases}
w_0 & \text{if $t \in \VV$} \\
\displaystyle w(f) + \smash[b]{\sum_{1 \leqslant i \leqslant n}}
\scoeff(f,i) \cdot w(t_i) & \text{if $t = f(\seq{t})$}
\end{cases}
\]

\smallskip
\noindent
The \emph{variable coefficient} $\vcoeff(x,t)$ of a variable $x$
in a term $t$ is inductively defined as follows:
\[
\vcoeff(x,t) = \begin{cases}
1 & \text{if $t = x$} \\
0 & \text{if $t \in \VV \setminus \{ x \}$} \\
\displaystyle \smash[b]{\sum_{1 \leqslant i \leqslant n}}
\scoeff(f,i) \cdot \vcoeff(x,t_i) & \text{if $t = f(\seq{t})$}
\end{cases}
\]

\begin{definition}
\label{def:s-ackbo}
The order $>_\actkbo$ is obtained from Definition~\ref{def:ackbo}
by replacing the condition
``\,$|s|_x \geqslant |t|_x$ for all $x \in \VV$\,'' with
``\,$\vcoeff(x,s) \geqslant \vcoeff(x,t)$ for all $x \in \VV$\,''
and using the modified weight function introduced above.
\end{definition}

In order to guarantee AC compatibility of $>_\actkbo$, the
subterm coefficient function $\scoeff$ has to assign the value $1$ to
arguments of AC symbols. This follows by considering the terms
$t \circ (u \circ v)$ and $(t \circ u) \circ v$ for an AC symbol
$\circ$ with $\scoeff({\circ},1) = m$ and $\scoeff({\circ},2) = n$.
We have
\begin{align*}
w(t \circ (u \circ v))
&= 2 \cdot w(\circ) + m \cdot w(t) + mn \cdot w(u) + n^2 \cdot w(v)
\\
w((t \circ u) \circ v)
&= 2 \cdot w(\circ) + m^2 \cdot w(t) + mn \cdot w(u) + n \cdot w(v)
\end{align*}
Since $w(t \circ (u \circ v)) = w((t \circ u) \circ v)$ must hold for
all possible terms
$t$, $u$, and $v$, it follows that $m = m^2$ and $n^2 = n$, implying
$m = n = 1$.%
\footnote{This condition is also obtained by restricting
\cite[Proposition 4]{BL87} to linear polynomials.}
The proof of the following theorem is very similar to the one of
Theorem~\ref{thm:ackbo} and hence omitted.

\begin{theorem}
If $\scoeff(f,1) = \scoeff(f,2) = 1$ for every function symbol
$f \in \FF_\AC$
then $>_\actkbo$ is an AC-compatible simplification order.
\qed
\end{theorem}

Subterm coefficients can be viewed as linear interpretations.
\citeN{L79} suggested to use polynomial interpretations for the weight
function of KBO.
A general framework for the use of arbitrary well-founded algebras in
connection with KBO is described in \cite{MZ97}. These developments
can be lifted to the AC setting with little effort.

\begin{example}
\label{example kusakari}
Consider the following TRS $\RR$ with $\circ \in \FF_\AC$:
\\[\abovedisplayskip]
\begin{minipage}{.45\textwidth}
\vspace{-\abovedisplayskip}
\begin{align}
 \m{f}(\m{0}, x \circ x) &\to x
 \tag{1} \\
 \m{f}(x,\m{s}(y)) &\to \m{f}(x \circ y, \m{0})
 \tag{2}
\end{align}
\end{minipage}
\hfill
\begin{minipage}{.5\textwidth}
\vspace{-\abovedisplayskip}
\begin{align}
 \m{f}(\m{s}(x),y) &\to \m{f}(x \circ y, \m{0})
 \tag{3} \\
 \m{f}(x \circ y, \m{0}) &\to \m{f}(x,\m{0}) \circ \m{f}(y,\m{0})
 \tag{4}
\end{align}
\end{minipage}
\\[\belowdisplayskip]
Termination of $\RR$ was shown using AC dependency pairs
in~\cite[Example~4.2.30]{K00}. Consider a precedence
$\m{f} > \circ > \m{s} > \m{0}$, and weights and subterm coefficients
given by $w_0 = 1$ and the following interpretation $\AA$, mapping
function symbols in $\FF$ to linear polynomials over $\Nat$:
\begin{xalignat*}{4}
 \m{s}_\AA(x) 	&= x + 6 &
 \m{f}_\AA(x,y) &= 4x + 4y + 5 &
 x \circ_{\!\AA} y &= x + y + 3 &
 \m{0}_\AA 	&= 1 
\end{xalignat*}
It is easy to check that the first three rules result in a weight
decrease. The left- and right-hand side of rule $(4)$ are both
interpreted as $4x+4y+21$, so both terms have weight $29$,
but since $\m{f} > \circ$ we conclude termination of $\RR$
from case~1 in Definition~\ref{def:ackbo} (\ref{def:s-ackbo}).
Note that termination of $\RR$ cannot be shown by AC-RPO or any of the
previously considered versions of AC-KBO.
\end{example}

\section{Experiments}
\label{experiments}

We ran experiments on a server equipped with eight dual-core AMD
Opteron$^{\mbox{\begin{scriptsize}\textregistered\end{scriptsize}}}$
processors 885 running at a clock rate of 2.6GHz with 64GB of main memory.
The different versions of AC-KBO considered in this paper as well
as AC-RPO~\cite{R02} were implemented on top of {\TTTT} using encodings
in SAT/SMT. These encodings resemble those for standard KBO~\cite{ZHM09}
and transfinite KBO~\cite{WZM12}.
The encoding of multiset extensions of order pairs are based on
\cite{CGST12}, but careful modifications were required
to deal with submultisets induced by the precedence.

\begin{table}[tb]
\caption{Experiments on 145 termination and 67 completion problems.}
\label{table}
\centering
\begin{tabular}{@{}%
l@{\qquad}
c@{~}c@{~}c@{}l@{\qquad}%
c@{~}c@{~}c@{}l@{\qquad}%
c@{~}c@{~}c@{}
}
\hline
\multicolumn{1}{c}{} &
\multicolumn{3}{c}{orientability} &&
\multicolumn{3}{c}{AC-DP} &&
\multicolumn{3}{c@{}}{completion}
\\
\cline{2-4}\cline{6-8}\cline{10-12}
method & 
yes & time & $\infty$ &&
yes & time & $\infty$ &&
yes & time & $\infty$
\\
\hline
AC-KBO &  
 32 & 1.7 & 0 &&
 66 & 463.1 & 3 &&
 25 & 2278.6 & 37
\\ 
Steinbach &   
23 & 1.6 & 0 &&
50 & 463.2 & 2 &&
24 & 2235.4 & 36
\\
Korovin \& Voronkov &
 30 & 2.0 & 0 &&
 66 & 474.3 & 4 &&
 25 & 2279.4 & 37
\\
KV$'$ &
 30 & 2.1 & 0 &&
 66 & 472.4 & 3 &&
 25 & 2279.6 & 37
\\
subterm coefficients &
 37 & \makebox[0mm][r]{4}7.1 & 0 &&
 68 & 464.7 & 2 &&
 28 & 1724.7 & 26
\\
AC-RPO &   
 63 & 2.8 & 0 &&
 79 & 501.5 & 4 &&
 28 & 1701.6 & 26
\\
\hline
\vphantom{\large $|$}%
total &
72 & & &&
94 & & &&
31 & &
\\
\hline
\end{tabular}
\end{table}

For termination experiments, our test set comprises all AC problems in
the \emph{Termination Problem Data Base~9.0},%
\footnote{\url{http://termination-portal.org/wiki/TPDB}}
all examples in this paper, some further problems harvested from the
literature, and constraint systems produced by the completion
tool {\mkbtt}~\cite{W13} (145 TRSs in total). The timeout was
set to 60 seconds.
The results are summarized in Table~\ref{table}, where we list for each
order the number of successful termination proofs, the total time, and the
number of timeouts (column $\infty$).  The `orientability' column directly
applies the order to orient all the rules.  Although AC-RPO succeeds on
more input problems, termination of 9
TRSs could only be established by (variants of) AC-KBO. 
We found that our definition of AC-KBO is about equally powerful 
as Korovin and Voronkov's order, but both are considerably 
more useful than Steinbach's version.
When it comes to proving termination, we did not
observe a difference between Definitions~\ref{def:KV} and~\ref{def:KV'}.
Subterm coefficients clearly increase the success rate,
although efficiency is affected.
In all settings partial precedences were allowed.

The `AC-DP' column applies the order in the AC-dependency pair framework of
\cite{ALM10}, in combination with \emph{argument filterings} and
\emph{usable rules}. Here AC symbols in dependency pairs are
\emph{unmarked}, as proposed in \cite{MU04}.
In this setting the variants of AC-KBO become
considerably more powerful and competitive to AC-RPO, since
argument filterings relax the variable condition,
as pointed out in \cite{ZHM09}.

For completion experiments, we ran the normalized completion tool 
{\mkbtt} with AC-RPO and the variants of AC-KBO for 
termination checks on 67 equational systems collected from the 
literature. The overall timeout was set to 60 seconds, the timeout for 
each termination check to 1.5 seconds. 
The `completion' column in
Table~\ref{table} summarizes our 
results, listing for each order the number of successful completions,
the total time, and the number of timeouts.
It should be noted that the results do not change if the overall timeout 
is increased to 600 seconds. For several of these input problems it is
actually unknown whether an AC-convergent system exists.

All experimental details, source code, and {\TTTT} binaries are 
available online.\footnote{%
\url{http://cl-informatik.uibk.ac.at/software/ackbo}}

The following example can be completed using AC-KBO, whereas AC-RPO 
does not succeed.

\begin{example}
Consider the following TRS~$\RR$~\cite{MU04} for addition of binary
numbers:
\begin{xalignat*}{3}
\# + \m{0} &\to \# &
x\m{0} + y\m{0} &\to (x + y)\m{0} &
x\m{1} + y\m{1} &\to (x + y + \#\m{1})\m{0} \\
x + \# &\to x &
x\m{0} + y\m{1} &\to (x + y)\m{1}
\end{xalignat*}
Here ${+} \in \FF_\AC$, $\m{0}$ and $\m{1}$ 
are unary operators in postfix notation, and $\#$ denotes the empty bit 
sequence. For example, $\#\m{100}$ represents the number 4.
This TRS is not compatible with AC-RPO but AC termination can 
easily be shown by AC-KBO, for instance with the weight function
$(w,w_0)$ with $w(\m{+}) = 0$, $w_0 = w(\m{0}) = w(\#) = 1$, and 
$w(\m{1}) = 3$.
It can be completed into an AC-convergent TRS using AC-KBO.
\end{example}

\section{Conclusion}
\label{conclusion}

We revisited the two variants of AC-compatible extensions of KBO.
We extended the first version $>_\steinbach$ introduced by Steinbach
\cite{S90} to a new version $>_\ackbo$, and presented a rigorous
correctness
proof. By this we conclude correctness of $>_\steinbach$, which had been
put in doubt in~\cite{KV03b}.
We also modified the order $>_\KV$ by \citeANP{KV03b}
to a new version $>_\KVC$ which is 
monotone on non-ground terms, in contrast to $>_\KV$.
We further presented several complexity results regarding these variants
(see Table~\ref{complexity results}).
While a polynomial time algorithm is known for the orientability problem
of standard KBO~\cite{KV03}, the problem becomes NP-complete even
for the ground version of $>_\KV$, as well as for our $>_\ackbo$.
Somewhat unexpectedly, even deciding $>_\KVC$ is NP-complete while
deciding standard KBO is linear \cite{L06b}.
In contrast, the membership problem is polynomial-time decidable
for our $>_\ackbo$.
Finally, we implemented these variants of AC-compatible KBO
as well as the AC-dependency pair framework of
\citeN{ALM10}.
We presented full experimental results both for termination proving
and normalized completion.

\begin{table}[t]
\caption{Complexity results (KV is the ground version of $>_\KV$).}
\label{complexity results}
\centering
\begin{tabular}{@{}l@{\quad}c@{~}c@{~}c@{~}c@{~}c@{~}c@{}}
\hline
\phantom{\large $|$}%
problem & KBO & S & AC-KBO & KV & KV$'$ & AC-RPO \\
\hline
\\[-2.2ex]
membership    & P & P & P & P & NP-complete & NP-hard \\[.5ex]
orientability & P & ? & NP-complete & NP-complete &
NP-complete & NP-hard \\
\hline
\end{tabular}
\end{table}

\paragraph{Acknowledgments.}
We are grateful to Konstantin Korovin for discussions and
the reviewers of the conference version~\cite{YWHM14}
for their detailed comments which helped to improve the presentation.
Ren\'e Thiemann suggested the proof of Lemma~\ref{ACKBO in NP}.

\bibliographystyle{acmtrans}
\bibliography{references}
\fi

\ifx\ARTICLE\undefined
\appendix

\section{Omitted Proofs}

\renewcommand{\thetheorem}{A.\arabic{theorem}}

\subsection{Correctness of $>_\ackbo$}
\label{A.1}

First we show that $({=_\AC},{>_\ackbo})$ is an order pair. To
facilitate the proof, we decompose $>_\ackbo$ into several orders. We
write
\begin{itemize}
\item
$s >_{01} t$ if $|s|_x \geqslant |t|_x$ for all $x \in \VV$
and either $w(s) > w(t)$ or $w(s) = w(t)$ and case~0 or
case~1 of Definition~\ref{def:ackbo} applies,
\smallskip
\item
$s >_{23,k} t$ if $|s|, |t| \leqslant k$, $|s|_x \geqslant |t|_x$
for all $x \in \VV$, $w(s) = w(t)$, and case~2 or case~3 applies.
\end{itemize}
The union of $>_{01}$ and $>_{23,k}$ is denoted by
$>_k$. The next lemma states straightforward properties.

\begin{lemma}
\label{lem:suborders}
The following statements hold:
\begin{enumerate}
\item
${>_\ackbo} \:=\: \bigcup\,\{ {>_k} \mid k \in \Nat \}$,
\item
$({=_\AC}, {>_{01}})$ is an order pair, and
\item
$({>_{01} \cdot >_k}) \cup ({>_k \cdot >_{01}})
\:\subseteq\: {>_{01}}$.
\end{enumerate}
\end{lemma}
\begin{proof*}
\begin{enumerate}
\item
The inclusion from right to left is obvious from the definition.
For the inclusion from left to right, suppose $s >_\ackbo t$.
If either $w(s) > w(t)$, or $w(s) = w(t)$ and case~0 or case~1 of
Definition~\ref{def:ackbo} applies,
then trivially $s >_{01} t$.
If case~2 or case~3
applies, then $s >_{23,k} t$ for
any $k$ with $k \geqslant \max(|s|,|t|)$.
\item
\smallskip
First we show that $>_{01}$ is transitive.  Suppose $s >_{01} t >_{01} u$.
If $w(s) > w(t)$ or $w(t) > w(u)$, then $w(s) > w(u)$ and $s >_{01} u$.
Hence suppose $w(s) = w(t) =  w(u)$.  Since $s, t \notin \VV$, we may write
$s = f(\seq{s})$ and $t = g(\seq[m]{t})$ with $f > g$.  Because of
admissibility, $g$ is not a unary symbol with $w(g) = 0$.
Thus $u \notin \VV$, and we may write $u = h(\seq[l]{u})$ with
$g > h$. By the transitivity of $>$ we obtain $s >_{01} u$.
The irreflexivity of $>_{01}$ is obvious from the definition.
It remains to show the compatibility condition
${=_\AC} \cdot {>_{01}} \cdot {=_\AC} \subseteq {>_{01}}$. This easily
follows from the fact that $w(s) = w(t)$ and $\rt(s) = \rt(t)$ whenever
$s =_\AC t$.
\item
\smallskip
Suppose $s = f(\seq{s}) >_{01} t = g(\seq[m]{t}) >_k u$. If
$t >_{01} u$ then $s >_{01} u$ follows from the transitivity of $>_{01}$.
Suppose $t >_{23,k} u$.  So $w(t) = w(u)$.
Thus $w(s) > w(u)$ if $w(s) > w(t)$,
and case~1 applies if $w(s) = w(t)$.
The inclusion ${>_k} \cdot {>_{01}} \subseteq {>_k}$ is proved in
exactly the same way.
\qed
\end{enumerate}
\end{proof*}

\begin{lemma}
\label{lem:f-extension}
Let $>$ be a precedence, $f \in \FF$, and
$({\GS},{\GT})$ an order pair on terms. Then
$({\GS^f},{\GT^f})$ is an order pair.
\end{lemma} 

\begin{proof}
We first prove compatibility. 
Suppose $S \GS^f T \GT^f U$.
From $T \GT^f U$ we infer that
$\rrs{T}{\nless} \uplus T{\restriction}_\VV \GT^\mul
\rrs{U}{\nless} \uplus U{\restriction}_\VV$.
Hence
$\rrs{S}{\nless} \GT^\mul \rrs{U}{\nless} \uplus U{\restriction}_\VV -
S{\restriction}_\VV$
follows from $S \GS^f T$.
Hence also
$S \mathrel{({\GS} \cdot {\GT})}^f U$.
We obtain the desired $S \GT^f U$ from
the compatibility of $\GS$ and $\GT$.
Transitivity of $\GS^f$ and $\GT^f$ is obtained in a very similar way.
Reflexivity of $\GS^f$ and irreflexivity of $\GT^f$ are obvious.
\end{proof}

We employ the following simple criterion to construct order pairs,
which enables us to prove correctness in a modular way.

\begin{lemma}
\label{lem:chain}
Let $({\GS},{\GT_k})$ be order pairs for $k \in \Nat$ with
${\GT_k} \subseteq {\GT_{k+1}}$. If $\GT$ is the union of all $\GT_k$ then
$({\GS},{\GT})$ is an order pair.
\end{lemma}

\begin{proof}
The relation $\GS$ is a preorder by assumption.
Suppose $s \GT t \GT u$.  By assumption
there exist $k$ and $l$ such that $s \GT_k t \GT_l u$. Let
$m = \max(k,l)$.  We obtain
$s \GT_m t \GT_m u$ from the assumptions of the lemma and hence $s \GT_m u$
follows from the fact that $({\GS},{\GT_m})$ is an order pair.
Compatibility is an immediate consequence of the assumptions
and the irreflexivity of $\GT$
is obtained by an easy induction proof.
\end{proof}

\begin{proof}[Proof of Lemma \ref{lem:ackbo order pair}]
According to Lemmata~\ref{lem:chain} and \ref{lem:suborders}(1), it is
sufficient to prove that $({=_\AC},{>_k})$ is an order pair for all
$k \in \Nat$. Due to Lemma~\ref{lem:suborders}(2,3) it suffices
to prove that $({=_\AC},{>_{23,k}})$ is an order pair,
which follows by using induction on $k$ in combination
with Lemma~\ref{lem:f-extension} and Theorem~\ref{thm:order pair}.
\end{proof}

\begin{proof}[Proof of Theorem \ref{thm:total}]
Let $\TT_k$ denote the set of ground terms of size at most $k$.  We use
induction on $k \geqslant 1$ to show that $>_\ackbo$ is AC-total on
$\TT_k$.  Let $s, t \in \TT_k$.  We consider the case where $w(s) = w(t)$
and $\rt(s) = \rt(t) = f \in \FF_\AC$.  The other cases follow as for
standard KBO.  Let $S = \tf{f}(s)$ and $T = \tf{f}(t)$.
Clearly $S$ and $T$ are multisets over $\TT_{k-1}$.
According to the induction hypothesis,
$>_\ackbo$ is AC-total on $\TT_{k-1}$ and since
multiset extension preserves AC totality,
$>_\ackbo^\mul$ is AC-total on multisets over $\TT_{k-1}$.
Hence for any pair of multisets $U$ and $V$ over $\TT_{k-1}$,
either
\begin{gather*}
U >_\ackbo^\mul V
\quad\text{or}\quad
V >_\ackbo^\mul U
\quad\text{or}\quad 
U =_\AC^\mul V
\end{gather*}
Because the precedence $>$ is total and $S$ and $T$ contain
neither variables nor terms with $f$ as their root symbol, we have
\begin{xalignat*}{2}
S &= \rrs{S}{\nless} \cup \rrs{S}{<} = \rrs{S}{>} \cup \rrs{S}{<} &
T &= \rrs{T}{\nless} \cup \rrs{T}{<} = \rrs{T}{>} \cup \rrs{T}{<}
\end{xalignat*}
If $\rrs{S}{>} >_\ackbo^\mul \rrs{T}{>}$ or
$\rrs{T}{>} >_\ackbo^\mul \rrs{S}{>}$
then case~3(a) of Definition~\ref{def:ackbo} is applicable to derive
either $s >_\ackbo t$ or $t >_\ackbo s$.
Otherwise we must have $\rrs{S}{>} =_\AC^\mul \rrs{T}{>}$ by
AC-totality. If $|S| > |T|$ then we obtain $s >_\ackbo t$ by case~3(b).
Similarly, $|S| < |T|$ gives rise to $t >_\ackbo s$.

In the remaining case we have both $\rrs{S}{>} =_\AC^\mul \rrs{T}{>}$
and $|S| = |T|$.
Using case~3(c)
of Definition~\ref{def:ackbo} we obtain
$s >_\ackbo t$ when $\rrs{S}{<} >_\ackbo^\mul \rrs{T}{<}$
and $t >_\ackbo s$ when $\rrs{T}{<} >_\ackbo^\mul \rrs{S}{<}$.
By AC totality there is one case remaining:
$\rrs{S}{<} =_\AC^\mul \rrs{T}{<}$. Combined with
$\rrs{S}{>} =_\AC^\mul \rrs{T}{>}$ we obtain
$S =_\AC^\mul T$. We may write
$S = \{ \seq[n]{s} \}$ and $T = \{ \seq[n]{t} \}$ such that
$s_i =_\AC t_i$ for all $1 \leqslant i \leqslant n$.
Since $f$ is an AC symbol,
$s =_\AC f(s_1,f(\dots, s_n)\dots)$ and 
$t =_\AC f(t_1,f(\dots, t_n)\dots)$, from which we
conclude $s =_\AC t$.
\end{proof}

\subsection{Correctness of $>_\KVC$}

We prove that $>_\KVC$ is an AC-compatible
simplification order. The proof mimics the one given in
Sections~\ref{AC-KBO} and~\ref{A.1} for $>_\ackbo$,
but there are some subtle differences.
The easy proof of the following lemma is omitted.

\begin{lemma}
\label{lem:Wt order pair}
The pairs $({=_\AC},{>_\Wt})$ and $({\geqslant_\kvc},{>_\Wt})$ are order
pairs.
\qed
\end{lemma}

\begin{lemma}
The pair $({=_\AC},{>_\KVC})$ is an order pair.
\end{lemma}
\begin{proof}
Similar to the proof of Lemma~\ref{lem:ackbo order pair}, except for
case~3 of Definition~\ref{def:KV'}, where we need
Lemma~\ref{lem:Wt order pair} and
Theorem~\ref{thm:order pair}.
\end{proof}

The subterm property follows exactly as in the proof of
Lemma~\ref{lem:subterm}; note that
the relation $>_{01}$ has the subterm property, and
we obviously have ${>_{01}} \subseteq {>_\KVC}$.

\begin{lemma}
\label{lem:kvd-sub}
The order $>_\KVC$ has the subterm property.
\qed
\end{lemma}

\begin{lemma}
\label{lem:kvd-mono}
The order $>_\KVC$ is closed under contexts.
\end{lemma}
\begin{proof*}
Suppose $s >_\KVC t$. We follow the proof for $>_\ackbo$ 
in Lemma~\ref{lem:ackbo monotone} and consider here the case that
$w(s) = w(t)$. We will show that one of the cases~3(a,b,c) in
Definition~\ref{def:KV'} (\ref{def:KV}) is
applicable to $S = \tf{h}(s)$ and $T = \tf{h}(t)$.
Let $f = \rt(s)$ and $g = \rt(t)$.
The proof proceeds by case splitting according to the  derivation of
$s >_\KVC t$.
\begin{itemize}
\item
Suppose $s = f^k(t)$ with $k > 0$ and $t \in \VV$.
Admissibility enforces $f > h$ and thus
$\rrs[h]{S}{\nless} = \{ s \} \geqslant_\kvc^\mul \{ t \}$.
We have $|S| = |T| = 1$ and $S >_\KVC^\mul T$.
Hence 3(c) applies. (This case breaks down for $>_\KV$.)
\smallskip
\item
Suppose $f = g \notin \FF_\AC$.
We have $S \geqslant_\kvc^\mul T$, $|S| = |T| = 1$, and
$S = \{ s \} >_\KVC^\mul \{ t \} = T$.
Hence 3(c) applies.
\smallskip
\item
The remaining cases are similar to the proof of
Lemma~\ref{lem:ackbo monotone}, except that we use Lemma~\ref{lem:tf}
with $({\geqslant_\kvc},{>_\Wt})$.
\qed
\end{itemize}
\end{proof*}

For closure under substitutions we need to extend 
Lemma~\ref{lem:f-key} with the following case:
\begin{enumerate}
\item[\textit{3.}]
\textit{If $S \GS^f T$ and $S' \NGT^f T'$ then
$S' - T' \supseteq S\sigma - T\sigma$ and 
$T\sigma - S\sigma \supseteq T' - S'$.}
\end{enumerate}

\begin{proof*}
We continue the proof of Lemma~\ref{lem:f-key}.
From $\tf{f}(U\sigma) = U\sigma$ we infer that
$T' = T{\restriction}_\FF\sigma \uplus U\sigma \uplus \tf{f}(X\sigma)$.
On the other hand,
$S' = S{\restriction}_\FF\sigma \uplus \tf{f}(Y\sigma) \uplus
\tf{f}(X\sigma)$ with $Y = S{\restriction}_\VV - X$.
Hence
\begin{align*}
T' - S'~ &\subseteq~
T{\restriction}_\FF\sigma \uplus U\sigma - S{\restriction}_\FF\sigma \\
&=~
T{\restriction}_\FF\sigma \uplus U\sigma \uplus X\sigma -
(S{\restriction}_\FF \uplus X\sigma) \\
&\subseteq~
T\sigma - S\sigma
\end{align*}
and
\begin{align*}
S' - T'~ &\supseteq~
S{\restriction}_\FF\sigma - T{\restriction}_\FF\sigma - U\sigma \\
&=~
S{\restriction}_\FF\sigma \uplus X\sigma -
(T{\restriction}_\FF \uplus U\sigma \uplus X\sigma) \\
&\supseteq~
S\sigma - T\sigma
\end{align*}
establishing the desired inclusions. \qed
\end{proof*}

\begin{lemma}
\label{kvd-stab}
The order $>_\KVC$ is closed under substitutions.
\end{lemma}

\begin{proof}
By induction on $|s|$ we verify that $s >_\KVC t$ implies 
$s\sigma >_\KVC t\sigma$.
If $s >_\KVC t$ is derived by one of the cases~0, 1, 2, 3(a) or 3(b)
in Definition~\ref{def:KV'}
(\ref{def:KV}),
the proof of Lemma~\ref{lem:ackbo monotone} goes through.
So suppose that $s >_\KVC t$ is derived by case~3(c)
and 
further suppose that
$s\sigma >_\KVC t\sigma$ can be derived neither by case~3(a) nor 3(b).
By definition we have
$\tf{f}(s) >_\KVC^\mul \tf{f}(t)$. This is equivalent\footnote{%
This property is well-known for standard multiset extensions (involving a
single \proper\ order). It is also not difficult to prove for the
multiset extension defined in Definition~\ref{lex and mul}.}
to
\[
\tf{f}(s) - \tf{f}(t) >_\KVC^\mul \tf{f}(t) - \tf{f}(s)
\]
We obtain $\tf{f}(s)\sigma - \tf{f}(t)\sigma >_\KVC^\mul
\tf{f}(t)\sigma - \tf{f}(s)\sigma$ from the induction hypothesis and thus
$\tf{f}(s\sigma) - \tf{f}(t\sigma) >_\KVC^\mul
\tf{f}(t\sigma) - \tf{f}(s\sigma)$ by 
Lemma~\ref{lem:f-key}(1).
Using the earlier equivalence, we infer
$\tf{f}(s\sigma) >_\KVC^\mul \tf{f}(t\sigma)$ and hence case~3(c)
applies to obtain the desired
$s\sigma >_\KVC t\sigma$.
\end{proof}

The combination of the above results proves Theorem~\ref{KV' correctness}.
%
%

\subsection{NP-Hardness of AC-KBO}
\label{sec:ACKBO NP-hard}

Next we show NP-hardness of the orientability problem for $>_\ackbo$.
To this end we introduce the TRS $\RR_0'$ consisting of the rules
\[
\High(p_1(\Bot)) \to p_1(\High(\Bot))
\qquad\cdots\qquad
\High(p_m(\Bot)) \to p_m(\High(\Bot))
\]
together with a rule
$\ToP_i^0(\ToP_i^1(\Bot)) \to \ToP_i^1(\ToP_i^0(\Bot))$
for each clause $C_i$ that contains a negative literal.
The next property is immediate.

\begin{lemma}
\label{lem:base2}
If $\RR_0' \subseteq {>_\ackbo}$ then
$\ToP_i^0 > \ToP_i^1$ for all $1 \leqslant i \leqslant n$ and
$\High > p_j$ for all $1 \leqslant j \leqslant m$.
\qed
\end{lemma}

The TRS $\RR_0 \cup \RR_0' \cup
\{ \ell_i \to r_i \mid 1 \leqslant i \leqslant n \}$
is denoted by $\RR_\phi'$.

\begin{lemma}
\label{lem:base3}
Suppose $\High > + > \Low$
and the consequence of Lemma~\ref{lem:base2} holds.
Then $\RR_\phi' \subseteq {>_\ackbo}$
for some $(w, w_0)$ if and only if for every $i$
there is some $p$ such that $p \in C_i$ with $p \nless +$ or
$\neg p \in C_i$ with $+ > p$.
\end{lemma}
\begin{proof}
The ``if'' direction is analogous to Lemma~\ref{lem:clause}.
Let us prove the ``only if'' direction by contradiction.
Suppose $+ > p'_j$ for all $1 \leqslant j \leqslant k$,
$p''_j \nless +$ for all $1 \leqslant j \leqslant l$,
and $\RR_\phi' \subseteq {>_\ackbo}$.
As discussed in the proof of Lemma~\ref{lem:clause},
for the multisets $V$ and $W$ 
on page~\pageref{encoding}
we obtain $V >_\ackbo^\mul W$ and
all terms in $V$ and $W$ have the same weight.
With the help of Lemma~\ref{lem:base2} we infer that
$\High(\ToP_i^0(\ToP_i^0(\Bot))) \in W$ is greater than every
other term in $V$ and $W$.
This contradicts $V >_\ackbo^\mul W$.
\end{proof}

Using Lemmata~\ref{lem:base2} and~\ref{lem:base3},
Theorem~\ref{ACKBO orientation NP-hard} can now be proved
in the same way as Theorem~\ref{KV orientation NP-hard}.

\subsection{AC-RPO}

\begin{proof*}[Proof of Lemma \ref{lem:acrpo}]
Because of totality of the precedence, $\rrs{S}{\not <}$ is
identified with $\rrs{S}{>}$ in the sequel.
First suppose $s >_\acrpo t$ holds by case~4.
We may assume that $>_\acrpo$ and $>_\acrpoo$ coincide on smaller
terms.
The conditions on $\embsm{f}$ are obviously the same. We distinguish
which case applies.
\begin{enumerate} 
\item[4(a)]
We have 
$\rrs{S}{>} >_\acrpo^\mul \rrs{T}{>} \uplus T{\restriction}_\VV -
S{\restriction}_\VV$ and thus both
$\rrs{S}{>} \uplus S{\restriction}_\VV \geqslant_\acrpo^\mul 
\rrs{T}{>} \uplus T{\restriction}_\VV$ and
$\rrs{S}{>} >_\acrpo^\mul \rrs{T}{>}$. So case~4$'$(a) is applicable.
\smallskip
\item[4(b)]
We have $|S| > |T|$ and $S =_\AC^f T$, i.e.,
$\rrs{S}{>} =_\AC^\mul \rrs{T}{>} \uplus T{\restriction}_\VV -
S{\restriction}_\VV$, and in particular 
$T{\restriction}_\VV \subseteq S{\restriction}_\VV$. Thus 
$\rrs{S}{>} \uplus S{\restriction}_\VV \geqslant_\acrpo^\mul 
\rrs{T}{>} \uplus T{\restriction}_\VV$ holds. 
Since $T{\restriction}_\VV \subseteq S{\restriction}_\VV$ and
$|S| > |T|$ imply $\#(S) > \#(T)$, case~4$'$(b) applies.
\smallskip
\item[4(c)]
We obtain $\rrs{S}{>} \uplus S{\restriction}_\VV \geqslant_\acrpo^\mul 
\rrs{T}{>} \uplus T{\restriction}_\VV$ as in case~4(b).
Together with $|S| = |T|$ this implies $\#(S) \geqslant \#(T)$.
As $S = \rrs{S}{>} \uplus S{\restriction}_\VV \uplus \rrs{S}{<}$
and similar for $T$, we obtain $S >_\acrpo^\mul T$ from the
assumption $\rrs{S}{<} >_\acrpo^\mul \rrs{T}{<}$. Hence case~4$'$(c)
is applicable.
\end{enumerate} 
Now let $s >_\acrpoo t$ by case~4$'$.
Again we assume that $>_\acrpo$ and $>_\acrpoo$ coincide on smaller
terms. We have
$\rrs{S}{>} \uplus S{\restriction}_\VV \geqslant_\acrpo^\mul 
\rrs{T}{>} \uplus T{\restriction}_\VV$ ($\ast$).
\begin{enumerate} 
\item[4$'$(a)]
We have $\rrs{S}{>} >_\acrpo^\mul \rrs{T}{>}$. 
Suppose  $S \not >_\acrpo^f T$, i.e., 
$\rrs{S}{>} >_\acrpo^\mul \rrs{T}{>} \uplus T{\restriction}_\VV -
S{\restriction}_\VV$ does not hold.
This is only possible if there is some variable
$x \in T{\restriction}_\VV - S{\restriction}_\VV$ for which there is no
term $s' \in \rrs{S}{>}$ with $s' >_\acrpo x$. This however contradicts
($\ast$), so $S >_\acrpo^f T$ holds and case~4(a)
applies.
\smallskip
\item[4$'$(b)]
If $\rrs{S}{>} >_\acrpo^\mul \rrs{T}{>}$ holds then case~4(a)
applies by
the reasoning in case~4$'$(a).
Otherwise, due to ($\ast$) we must have 
$S =_\AC^f T$. Since $\#(S) > \#(T)$ implies $|S| > |T|$, case~4(b)
applies.
\smallskip
\item[4$'$(c)]
If $\#(S) > \#(T)$ is satisfied we argue as in the preceding case.
Otherwise $\#(S) \geqslant \#(T)$ and $\#(S) \ngtr \#(T)$. This implies
both $|S| = |T|$ and $S{\restriction}_\VV \supseteq T{\restriction}_\VV$.
We obtain $S =_\AC^f T$ as in case~4$'$(b).
From the assumption $S >_\acrpo^\mul T$ we infer
$\rrs{S}{<} >_\acrpo^\mul \rrs{T}{<}$ and thus case~4(c)
applies.
\qed
\end{enumerate}
\end{proof*}
\fi

\end{document}